\documentclass[11pt]{article}

\usepackage[overload]{empheq}
\usepackage{latexsym,dsfont,textcomp,amsmath}
\usepackage[english]{babel}
\usepackage[latin1]{inputenc}
\usepackage{enumerate,indentfirst}
\usepackage[colorlinks,citecolor=blue,urlcolor=blue]{hyperref}
\usepackage{amssymb}
\usepackage{amsthm}
\usepackage{hyperref}
\usepackage{graphicx}

\newtheorem{Theorem}{Theorem}[part]
\newtheorem{Definition}{Definition}[part]
\newtheorem{Proposition}{Proposition}[part]
\newtheorem{Assumption}{Assumption}[part]
\newtheorem{Lemma}{Lemma}[part]

\newtheorem{Remark}{Remark}[part]

\newcommand{\nc}{\newcommand}
\nc{\esssup}{\mathop{\mathrm{ess\,sup}}}
\nc{\essinf}{\mathop{\mathrm{ess\,inf}}}
\nc{\argmax}{\mathop{\mathrm{arg\,max}}}

\def \P{\mathbb{P}}

\def \R{\mathbb{R}}

\def \E{\mathbb{E}}
\def \F{\mathbb{F}}
\def \G{\mathbb{G}}
\def \Q{\mathbb{Q}}

\def \1{\mathds{1}}

\def \Ac{{\cal A}}

\def \Fc{{\cal F}}
\def \Gc{{\cal G}}
\def \Hc{{\cal H}}

\def \Mc{{\cal M}}

\def \l({{\left (}}
\def \r){{\right )}}
\def \l[{{\left [}}
\def \r]({{\right ]}}

\newcommand{\MBFigure}[6]{
$\left. \right.$ \\
\refstepcounter{figure}
\addcontentsline{lof}{figure}{\numberline{\thefigure}{\ignorespaces #5}}
\begin{center}
\begin{minipage}{#1cm}
\centerline{\includegraphics[width=#2cm,angle=#3]{#4}}
\begin{center}
\upshape{F\textsc{ig} \normal
\end{center}
size{\thefigure}. $-$} #5
\end{center}
\label{#6}
\end{minipage}
\end{center}
$\left. \right.$ \\}

\begin{document}

\title{Optimization problem and mean variance hedging on defaultable claims.}

\author{St\'ephane GOUTTE \thanks{Laboratoire de Probabilit\'es et Mod\`eles Al\'eatoires,
CNRS, UMR 7599, Universit{\'e}s Paris 7 Diderot.}   \thanks{Supported by the FUI project $R=MC^2$.
Mail: goutte@math.univ-paris-diderot.fr}
 {\sc,}\ Armand NGOUPEYOU $^*$\thanks{Supported by ALMA Research. Mail: armand.ngoupeyou@univ-paris-diderot.fr}
}

\maketitle
\begin{abstract}
We study the pricing and the hedging of claim $\psi$ which depends on the default times of two firms A and B. In fact, we assume that, in the market, we can not buy or sell any defaultable bond of the firm B but we can only trade defaultable bond of the firm A. 
 Our aim is then to find the best price and hedging of $\psi$ using only bond of the firm A. Hence, we solve this problem in two cases: firstly in a Markov framework using indifference price and solving a system of Hamilton-Jacobi-Bellman equations, secondly, in a more general framework, using the mean variance hedging approach and solving backward stochastic differential equations (BSDE).  
\end{abstract}

\vspace{0.5cm}

\textbf{MSC Classification (2010):} 60G48 60H10 91G40 49L20 \\
\section*{Introduction}
Models for pricing and hedging defaultable claim have generated a large debates between academics and practitioners during the last subprime crisis.  The challenge is to model the expected losses of derivatives portfolio by taking into account the counterparties 
 defaults. Indeed, they have been affected by the crisis and their agreement on the derivatives contracts can potentially vanish. In the literature, models for pricing defaultable securities have been initiated by Merton \cite{Mer74}. His approach consists of explicitly linking the default risk of a firm to its value. This model is a good issue to understand the default risk. However, it is less useful in practical applications since it is too difficult to capture the dynamics of the firm's value which depends on many macroeconomics factors. In response of these difficulties, Duffie and Singleton \cite{ DuSin03} introduced the reduced form modeling which has been followed by Madan and Unal \cite{MaUn98}, Jeanblanc and Rutkowski \cite{ JeRu02} and others. In this approach, the main tool is the "default intensity process", which describes in short terms the instantaneous probability of default. This process combined with the recovery rate of the firm represents the main tools necessary to manage the default risk. However, we should manage the default risk considering the financial market as a network where every default can affect another one and the propagation spreads as far as the connections exist. In the literature, to deal with this correlation risk, the most popular approach is the copula. This approach consists to define the joint distribution of the firms on the financial network with respect to the marginal distribution of each firm. In static framework \footnote{The framework where we don't consider the evolution of the survey probability given a filtration.}, Li \cite{Li00} was the first to develop this approach to model the joint distribution of the default times. But, all computations are done without considering the evolution of the survey probability given the available information. Thus, we can not describe the dynamics of the derivatives portfolio in this framework. In response of these limits on the static copula approach, El Karoui, Jeanblanc and Ying developped a conditional density approach \cite{ElkJeJia09}. An important point, in this framework, is that given this density, we can compute explicitly the default intensity processes of firms in the financial market considered. We will follow this approach and work without losing any generality in the explicit case where the financial network is defined only with two firms denoted by A and B.
The intensity process jumps when any default occurs. This jump impacts the default of the firm and makes some correlation between them. We assume that we can not buy or sell any defaultable bond of the firm B but we can trade a defaultable bond of the firm A. We will consider two different cases for pricing and hedging a general defaultable claims $\psi$: the indifference pricing in Markov framework and the Mean-Variance hedging approach for the general case. 

In the first case, we will work in a Markov framework. Our aim will be to find, using the correlation between the two firms, the indifference price of any contingent claim given the risk aversion. This risk aversion will be defined by an exponential utility function. We will express the indifference pricing as an optimization problem (see El Karoui and Rouge \cite{ElkRou00}) and we will use the Kramkov and Schachermayer \cite{KS} dual approach. Then solving this dual problem, we will find the solution of the indifference price. Moreover, the characterization of the optimal probability for the dual optimization problem will be solved by Hamilton-Jacobi-Bellman (HJB) equations since the defaultable bond price will be assumed to be a Markov process in this framework. We will also find an explicit formula for the optimal strategy. 
\par In a second case, we will be interested in a hedging problem using the Mean-Variance approach. We will assume that we work in a general setting (not necessarily Markov), then we will not be able to others use the HJB equations to characterize the corresponding value function. Hence, we will adopt the Mean Variance approach which has been introduced by Schweizer in \cite{Sch92} and generalized by many (\cite{Sch96}, \cite{Pham1}, \cite{Pham2}, \cite{DMSSS97}, \cite{Arai05}, \cite{LimA04}, \cite{Goutte}). Most of theses papers use martingales techniques and an important quantity, in this context, is the Variance Optimal Martingale Measure (VOM). The VOM, $\mathbb{\bar P}$, is the solution of the dual problem of minimizing the $L^2$-norm of the density $d\mathbb{Q}/d\mathbb{P}$, over all (signed) local martingale measures $\mathbb{Q}$ for the defaultable bond price of the firm A. If we consider the case of no jump dues to default, then the bond price process of the firm A is continuous. In this case, Delbaen and Schachermayer in \cite{DS96} proved the existence of an equivalent VOM $\mathbb{\bar P}$ with respect to $\mathbb{P}$. Moreover the price of any contingent claim $\psi$ is given by $\mathbb{E}^{\mathbb{\bar P}}(\psi)$. In Laurent and Pham \cite{Pham2}, they found an explicit characterization of the variance optimal martingale measure in terms of the value function of a suitable stochastic control problem. In the discontinuous case, when the so-called Mean-Variance Trade-off process (MVT) is deterministic, Arai \cite{Arai05} proved the same results. Since we will work in discontinuous case and since in our case the Mean variance trade-off process is not deterministic (due to the stochastic default intensity process), we will not be able to apply the standards results. Hence our work will be first to characterize the value process of the Mean-Variance problem and then to make some links with the existence and the characterization of the VOM. However, we will not really need to prove and assume this existence to solve our problem. Indeed, we will solve a system of quadratic Backward Stochastic Differential Equations (BSDE) and we will characterize the solution of the problem using BSDE's solutions. The main contribution in this part will be the explicit characterization of the BSDE's solutions without using the existence of the VOM. We will obtain an explicit representation of each coefficients of quadratic backward stochastic differential equations with respect to the parameters asset of our model. In particular, the main BSDE coefficient will follow a quadratic growth and its solution will be found in a constrained space. In a particular discontinuous filtration framework (where the asset parameters do not depend on the filtration generated by the jump), Lim \cite{LimA09} have reduced this constrained quadratic BSDE with jumps to a constrained quadratic BSDE without jump and solved the corresponding BSDE. In the discontinuous filtration due to defaults events, we will can not do the same assumption since the intensity processes depend on the jumps (the default events). Hence, using Kharroubi and Lim \cite{KhaLim11} technics, we will split the BSDE's with jumps into many continuous BSDEs with quadratic growth and we will conclude the existence of the solution using the standard results of Kobylanski \cite{Kob00}. 
\par\medskip
Hence, the paper is structured as follow, in a first section, we will give some notations and present our model with some results relative to credit risk modeling. Then, in a second part, we will study the case of pricing and hedging defaultable contingent claim in a Markov framework using indifference pricing. Then in the last section, we will study the pricing and hedging problems in a more general framework (not Markov) using mean variance hedging approach and solving a system of quadratic BSDEs.

\section{The defaultable model} 
In the sequel, we will work in the same model construction as in Bielecki and al. in \cite{BookMonique} chapter 4.
Let $T>0$ be a fixed maturity time and denote by $(\Omega,\F:=(\Fc_t)_{[0,T]},\P)$ an underlying probability space. The filtration $\F$ is generated by a one dimensional Brownian motion $\widetilde W$. Let $\tau^A$ and $\tau^B$ be the two default times of two firms A and B. Let define, for all $t\in [0,T]$:
\begin{equation}
H^A_t=1_{\{\tau^{A}\leq t\}} \quad \textrm{and }\quad H^B_t=1_{\{\tau^{B}\leq t\}}.
\end{equation}
  We define now some useful filtrations and definitions:
 \begin{equation*}
 {\cal G}^A_t={\cal F}_t\vee {\cal H}^B_t,\quad  \quad {\cal G}^B_t={\cal F}_t\vee {\cal H}^A_t \quad and \quad \Gc_t=\Fc_t \vee \Hc^A_t \vee \Hc^B_t,
\end{equation*}
 where $\Hc^A$ (resp. $\Hc^B$) is the natural filtration generated by $H^A$ (resp. $H^B$). We will denote by $\G:=\left(\Gc_t\right)_{t\in [0,T]}$, $\G^A:=\left(\Gc^A_t\right)_{t\in [0,T]}$ and $\G^B:=\left(\Gc^B_t\right)_{t\in [0,T]}$.
 \begin{Definition}[Initial time] Let  $\displaystyle \eta $ be a positive finite measure on $\displaystyle \R^2$. The  random times $\tau^A$ and $\tau^B$  are called initial times if,  for each $ t \in [0,T]$, their joint conditional law given $\displaystyle {\cal F}_t $ is absolutely continuous with respect to $\displaystyle \eta$.  Therefore,  there exists a  positive  family $\left(g_t (y)  \right)_{t \in [0, T]}$ of $\mathbb{F}$-martingales such that
 \begin{equation}\label{conditionnnel}
G_t(\theta^A,\theta^B)=\P(\tau^A>\theta^A,\tau^B>\theta^B\vert {\cal F}_t)=\int_{\theta^A}^{+\infty}\int_{\theta^B}^{+\infty} g_t(y_1,y_2)\eta(dy_1,dy_2),
\end{equation}
for each $ \theta^A,\theta^B \in \mathbb{R}^+ $ and $ t \in [0,T]$.
\end{Definition}
Regarding this definition, we make the following assumptions:
\begin{Assumption}\label{Hyptimes}( Properties of the default times)
\begin{itemize}
\item Processes $H^A$ and $H^B$ have no common jump: $\P\left(\tau_A=\tau_B\right)=0$. 
\item The default times $\tau_A$ and $\tau_b$ are initial times.
\end{itemize} 
\end{Assumption}
\noindent Hence, point 2. of the previous Assumption implies that the default times of firm $A$ and $B$ are correlated regarding our joint probability density $g_t$ (appearing in \eqref{conditionnnel}).
We now give a representation Theorem of our defaultable model.
\begin{Theorem} \label{ThmReprez}(Representation Theorem) Under Assumption \ref{Hyptimes}, for $i\in \{A,B\}$, there exists a positive $\mathbb{G}$-adpated process $\lambda^i$, called the $\mathbb{P}$-intensity of $H^i$, such that the process $M^i$ defined by
$$M^i_t=H^i_t-\int_0^t \lambda^i_s ds,$$
is a $\mathbb{G}$-martingale. Moreover, any local martingale $\zeta={(\zeta_t)}_{t\ge 0}$ admits the following decomposition: $\P$-a.s,
\begin{equation}\label{representation}
 \zeta_t=\zeta_0+\int_0^t  Z_s dW_s+
\int_0^t U^A_s dM^A_s+\int_0^t U^B_s dM^B_s,  \quad \forall \,  t\ge 0,
\end{equation}
where $Z, U^A$ and $U^B$ are $\G$-predictable processes and $W$ is the martingale part of the $\mathbb{G}$-semimartingale $\widetilde W$ in the enlarged filtration (see \cite{cam} for more details about the progressive enlargement of filtration and the characterization of the decomposition of any $\mathbb{F}$-semimartingale in the enlarged filtration $\mathbb{G}$). 
\end{Theorem}
\begin{proof}
The processes $\lambda^A$ and $\lambda^B$ are given explicitly since we assume that $\tau^A$ and $\tau^B$ are initial times and given our conditional law $G$. Moreover in Proposition 1.29, p54 in \cite{Ben11}, the author follows the proof of the representation Theorem of Kusuoka (representation theorem when the default times are independent of the filtration $\mathbb{F}$) to construct the proof when default times are initial. 
\end{proof}
\subsection{Dynamic of the Bond}
In our model, the traded asset will be the defaultable bond $D^A$ of the firm A. Using the decomposition \eqref{representation}, we represent the dynamics of this defaultable bond in the enlarged filtration $\G$ as in Corollary 5.3.2 of \cite{BookMonique}:
\begin{equation}\label{bondA}
{dD^A_t\over D^A_{t^-}} =\mu_tdt+\sigma^A_t dM^A_t+\sigma^B_t dM^B_t+\sigma_t dW_t,
\end{equation}
where $\mu, \sigma^A, \sigma^B$ and $\sigma$  are $\G$-predictable bounded processes. 
Therefore, given an initial wealth $x\geq 0$, if we assume that investors follow an admissible strategies $\pi$, which is represented by a set ${\cal A}$ of predictable processes $\pi$ such that 
\begin{equation}\label{1.1}
\E\left[\int_0^T \pi^2_s ds\right] <+\infty.
\end{equation}
Then we can define the dynamics of the wealth process, started with an initial wealth $x$ at time $t=0$ and following a strategy $\pi$, $X^{x,\pi}$ based on the trading asset $D^A$ by 
\begin{equation}\label{EqDynWea}
d X^{x,\pi}_t=\pi_t{dD^A_t\over D^A_{t^-}}=\pi_t\left[\mu_tdt+\sigma^A_t dM^A_t+ \sigma ^B_t dM^B_t+\sigma_t dW_t\right].
\end{equation}
Note that since all the coefficients in the dynamics of the wealth process are bounded, then for any $\pi \in {\mathcal A}$ we have that \eqref{1.1} implies:
$$\mathbb{E}\left[\sup_{0\le t\le T}{|X^{x,\pi}_t|}^2\right]<+\infty.$$
\subsection{The Defaultable claim}
We now introduce the concept of defaultable claim and give some explicit examples.
\begin{Definition}\label{DefClaim}
A generic defaultable claim $\psi$ with maturity $T>0$ on two firms A and B is defined as a vector 
$$\psi:=(X^A,X^B,Z^A,Z^B,\tau^A,\tau^B)$$
 with maturity T such that:
\begin{itemize}
\item $\tau^i$, $i\in \{A,B\}$ is the default time specifying the random time of default of the firm $i$ and thus also the default events $\{\tau^i \leq t\}$ for every $t\in [0,T]$. It is always assumed that $\tau^i$ is strictly positive with probability 1.
\item $X^A$ is the promised payoff which represents the random payoff received by the owner of the claim $\psi$ at time T, if there was no default of firm A prior to or at time T.
\item $X^B$ is the promised payoff which represents the random payoff received by the owner of the claim $\psi$ at time T, if there was no default of firm B prior to or at time T.
\item $Z^i$, $i\in \{A,B\}$is the recovery process which specifies the recovery payoff $Z_{\tau^i}$ received by the owner of a claim at time of default of the firm i, provided that the default occurs prior to or at maturity date T.
\end{itemize}
We can introduce now the payoff at time $T$ of this defaultable claim, which represents all cash flows associated with $(X^A,X^B,Z^A,Z^B,\tau^A,\tau^B)$. We will use also the notation $\psi$ for this payoff. Formally, the payoff process $\psi$ is defined through the formula by
\begin{eqnarray}\label{EqDiv}
\psi &=&X^A1_{\{\tau^A> T\}}+X^B1_{\{\tau^B> T\}}+\int_0^TZ^A_sdH^A_s+\int_0^TZ^B_sdH^B_s.
\end{eqnarray}
\end{Definition}
As an example, we can have a defaultable claim which only gives a terminal payoff of $H^1$ if no default occurs before time T. Hence, we will not receive money if one of the firms makes default. So our defaultable claim is given by
$$
\psi=H^1_{\{\tau^A \vee \tau^B >T\}}.
$$
We can also have a defaultable claim which gives an amount of money with respect to the default time of the firm B and gives a recovery amount $H^3$ if the firm A makes default
$$
\psi=H^1_{\{\tau^B>T\}}+H^2_{\{\tau^B\leq T\}}+\int_0^TH^3dH^A_s.
$$
\section{Hedging defaultable claim in Markov framework}\label{SectionHedging}
Let consider $\psi\in \Gc_T$ a bounded defaultable claim as defined in Definition \ref{DefClaim}, which depends on the default times $\tau^A$ of the firm A and $\tau^B$ of the firm B. Our aim is to find the best hedging and pricing of $\psi$ with respect to these defaults times.
\begin{Assumption}\label{HypUtili}
We assume that  $\mu, \sigma^A, \sigma^B,\sigma$  and the intensity processes $\lambda^A,\lambda^B$ are deterministic bounded functions of time, $H^A$ and $H^B$.
\end{Assumption}
\begin{Remark}\label{HypMarkov}
Under Assumption \ref{HypUtili}, we have that $(D^A, H^A ,H^B)$ is a Markov process.
\end{Remark}
\noindent We assume that the risk aversion of the investors is given by an exponential utility function $U$ with parameter $\delta$, given by
$$
U(x)=-\exp(-\delta x).
$$
Therefore, to define the indifference price or the hedging of $\psi$, we should solve the following equation:
$$
u^\psi(x+p)=u^0(x),
$$
where functions $u^\psi$ and $u^0$ are defined by:
\begin{equation}\label{optipb}
u^\psi(x)=\sup_{\pi \in {\cal A}}\mathbb{E}\left[-\exp(-\delta(X^{x,\pi}_T-\psi))\right] \quad \textrm{and} \quad u^0(x)=\sup_{\pi \in {\cal A}}\mathbb{E}\left[-\exp(-\delta X^{x,\pi}_T)\right].
\end{equation}
\subsection{ The dual optimization formulation}
To deal with the problems $\eqref{optipb}$, we use the duality theory developed by Kramkov and Schachermayer in \cite{KS}. In fact this theory allows us to find the optimal wealth at the horizon time T and the optimal risk-neutral probability $\Q^*$. 
Let recall now some results about the dual theory.
\begin{Theorem}\label{DeSch}[Kramkov and Schachermayer, Theorem 2.1 of \cite{KS}]\\
 Let $U$ be an utility function which satisfies the standards assumptions and consider the optimization problem: $ u(x)=\sup_{\pi \in {\cal A}} \mathbb{E}\left[U(X^{x,\pi}_T)\right]$,
then the dual function of $u$ defined by:
\begin{equation*}
v(y)=\sup_{x>0}\{u(x)-xy\},\quad u(x)=\inf_{y>0}\{v(y)+yx\}
\end{equation*}
is given by 
\begin{equation}\label{EqDelta}
v(y)=\inf_{\mathbb{Q} \in \mathcal{M}^e }\mathbb{E} \left(V\left[y {d\Q\over d\mathbb{P}}\right]\right),
\end{equation}
 \noindent where $V$ represents the dual function of $U$ and $\mathcal{M}^e$ represents the set of all risk-neutral probability measures.\\
Moreover, there exists an optimal martingale measure $\mathbb{Q}^*$ which solves the dual problem and  we have that the optimal wealth at time $T$ is given by:
$$ X^{x,\pi^*}_T=I\left[\nu Z^{\mathbb{Q}^*}_T\right], \hbox{ where } \nu \hbox{ is defined s.t.} \quad \mathbb{E}^{\mathbb{Q}^*}\left[ X^{x,\pi^*}_T\right]=x.$$
\noindent The function $I$ represents the inverse function of $U'$ and $Z^{\mathbb{Q}^*}_T$ represents the Radon Nikodym density of $\mathbb{Q}^*$ with respect to $\P$ on $\mathcal {G}_T$.
\end{Theorem}
\noindent We can apply this result to solve our optimization problem \eqref{optipb}. We will resolve only the case $\psi\not=0$. Indeed the particular case $\psi=0$ could be obtained as a particular case of these results. We obtain an analogous result of Delbaen and al. Theorem 2 in \cite{DGRSSS}, given by the following proposition:
\begin{Proposition} Let $\mathbb{Q}^*$ be the optimal risk-neutral probability which solves the dual problem 
\begin{equation}\label{dualopti}
\inf_{\mathbb{Q}\in {\mathcal M}^e}\left[ H(\Q \vert \mathbb{P})-\delta \mathbb{E}^{\mathbb{Q}} (\psi)\right]
\end{equation}
 then the optimal strategy $\pi^*\in {\mathcal A}$ solution of the optimization problem $\eqref{optipb}$ satisfies:
\begin{equation}\label{EqSolD1}
-{1\over \delta} \ln\left(Z^{\mathbb{Q}^*}_T\right)+\psi=x+{1\over \delta}\ln\left({y\over \delta}\right)+\int_0^T \pi^*_t dD^A_t,
\end{equation}
\noindent where $H(\mathbb{Q}\vert \mathbb{ P})$ represents the entropy of $\mathbb{Q}$ with respect to $\mathbb{P}$ $\left(i.e.\quad \E^\Q\left[\log\left(\frac{d\Q}{d\P}\right)\right]\right)$ and $y$ is a non negative constant. 
\end{Proposition}
\begin{proof}The proof is based on the Theorem $\ref{DeSch}$. Firstly, to match with assumptions of this theorem in the case $\psi\not=0$, we change the historical probability. Let define 
\begin{equation*}
{d\mathbb{P}^\psi\over d \mathbb{P}}{\Big \vert }_{{\mathcal G}_T}={\exp(\delta \psi)\over \mathbb{E}\left[\exp(\delta \psi)\right]}\quad \hbox{ and }\quad \widetilde u^\psi(x)=\sup_{\pi\in {\mathcal A}} \mathbb{E}^\psi\left[-\exp(-\delta X^{x,\pi}_T)\right], 
\end{equation*}
then setting $c=\mathbb{E}\left[\exp(\delta \psi)\right]$, we get
\begin{eqnarray*}
u^\psi(x)&=&\sup_{\pi \in {\cal A}}\mathbb{E}\left[-\exp(-\delta(X^{x,\pi}_T-\psi))\right]=\sup_{\pi \in {\cal A}}\mathbb{E}^{\P^\psi}\left[-c\exp(-\delta X^{x,\pi}_T)\right]\\
&=&\sup_{\pi \in {\cal A}}\mathbb{E}^{\P^\psi}\left[\exp\left(-\delta\left( -\frac{1}{\delta}\log(c)+X^{x,\pi}_T\right)\right)\right] =\sup_{\pi \in {\cal A}}\mathbb{E}^{\P^\psi}\left[\exp\left(-\delta X^{x-\frac{1}{\delta}\log(c),\pi}_T\right)\right].
\end{eqnarray*}
Hence by the definition of $\widetilde u^\psi(x)$, we obtain that $\widetilde u^\psi\left(x-{1\over \delta}\ln(c)\right)=u^\psi(x)$.
Then using the Theorem $\ref{DeSch}$, the dual function of $\widetilde u^\psi$ is given, for all $y>0$, by :
\begin{equation}\label{EqDuaV1}
\widetilde v^\psi(y)=\inf_{\Q \in \mathcal{M}^e }\mathbb{E} \left[V\left(y {d\Q \over d\mathbb{P}^\psi}\right)\right],
\end{equation}
where
$$
V(y)=\sup_{x>0}\{U(x)-xy\}=\sup_{x>0}\{-\exp(-\delta x)-xy\}={y\over \delta}\left[\ln\left({y\over \delta}\right)-1\right].
$$
Using this expression into \eqref{EqDuaV1} gives, after straightforward calculation, an explicit expression of the dual function which is given by
\begin{equation*}
\widetilde v^\psi(y)=V(y)+{y\over \delta}\ln(c)+{y\over \delta}\inf_{\mathbb{Q}\in {\mathcal M}^e}\left[ H(\mathbb{Q}\vert \mathbb{P})-\delta \mathbb{E}^{\mathbb{Q}} (\psi)\right].
\end{equation*}
Since  $\mathbb{Q}^*$ is the optimal risk-neutral probability which is solution of $\eqref{dualopti}$, 
 we deduce that the optimal wealth at time $T$ of the optimization problem $\eqref{optipb}$ is given by 
$$
X^{x,\pi^*}_T=I\left[y\frac{Z^{\mathbb{Q}^*}_T}{ Z^{\mathbb{Q}^\psi}_T}\right],
$$
 where $y$ is defined such that $\mathbb{E}^{\mathbb{Q}^*}\left[X^{x,\pi^*}_T\right]=x-{1\over \delta}\ln(c)$ and I is equal to $-V^{'}$.

Moreover from Owen in \cite{Owen02}, we can deduce that there exists an optimal strategy $\pi^* \in {\mathcal A}$ such that:
$$ 
X^{x,\pi^*}_T=I\left[y\frac{Z^{\mathbb{Q}^*}_T}{ Z^{\mathbb{Q}^\psi}_T}\right]=x-{1\over \delta}\ln(c)+\int_0^T \pi^*_t dD^A_t.
$$ 
In our case, since we work under the case of an exponential utility function with parameter $\delta$, we have that
$$
I(y):=-{1\over \delta}\ln\left({y\over \delta}\right).
$$
We finally get
$$x-{1\over \delta}\ln(c)+\int_0^T \pi^*_tdD^A_t=-{1\over \delta}\ln\left({y\over \delta}\right)-{1\over \delta}\log\left(Z^{\mathbb{Q}^*}_T\right)+\psi-{1\over \delta}\ln(c),$$
which concludes the proof of this proposition.
\end{proof}
\subsection{ Value function of the dual problem }
We are now interested in solving the dual problem. Firstly, let us consider the same problem with a different set of probability measures like ${\mathcal M}^e ={\mathcal Q}$, where ${\mathcal Q}$ represents the set of all probability measures $\mathbb{Q}\ll\mathbb{P}$. Then the value function is given by the entropy of $\psi$ with a parameter $\delta$.
But since we work in a more restricted set of probability measures  ${\mathcal M}^e$ which represents the set of all risk-neutral probabilities, the value function is then more difficult to precise. Indeed, to characterize the value function, we begin by defining the set ${\mathcal M}^e$.
Hence, let $\mathbb{Q}\in {\mathcal M}^e$ and define  $Z^{\mathbb Q}_T$ to be the Radon Nikodym density of $\mathbb{Q}$ with respect to $\mathbb{P}$. Considering the non negative  martingale process $Z^{\mathbb{Q}}_t=\mathbb{E}\left[Z^{\mathbb{Q}}_T\vert {\mathcal G}_t\right]$ and using representation Theorem \ref{ThmReprez} imply that there exists predictable processes $\rho^A$ and $\rho^B$ which take their values in ${\mathcal C}=(-1,+\infty)$  and a predictable process $\rho$ which takes its values in $\mathbb{R}$ such that for all $t \in [0,T]$,
$$dZ^{\mathbb{Q}}_t=Z^{\mathbb{Q}}_{t^-}\left( \rho^A_tdM^A_t+\rho^B_tdM^B_t+\rho_tdW_t\right).$$
Since $\mathbb{Q}$ is in $\Mc^e$, it is a risk-neutral probability, then  $ZD^A$ is a local martingale. This implies by Ito's calculus the following equation:
\begin{equation}\label{constraint}
\mu_t+\rho^A_t\sigma^A_t\lambda^A_t+\rho^B_t\sigma^B_t\lambda^B_t+\rho_t\sigma_t=0.
\end{equation}  
\begin{Remark}
We notice that the process $\rho$ depends explicitly on the values of $\rho^A$ and $\rho^B$.
\end{Remark}
\noindent Therefore using equation $\eqref{constraint}$, $\eqref{dualopti}$ can be view as find $\rho^A$ and $\rho^B$ which minimize:
\begin{equation}\label{pbdual2}
\inf_{\mathbb{Q}\in {\mathcal M}^e}\mathbb{E}^{\mathbb{Q}}\left[\ln(Z^{\mathbb{Q}}_T)-\delta \psi\right].
\end{equation} 
This is the dual problem we would like to solve. We make now an assumption on the decomposition form of our defaultable claim $\psi$.
\begin{Assumption}\label{HypClaim} The defaultable claim $\psi \in {\mathcal G}_T$ is given by 
$$
\psi=g(D^A_T)\textbf{1}_{\{\tau^B> T\}}+f(D^A_{{\tau^B}^-})\textbf{1}_{\{\tau^B\le T\}},
$$
 where g and f are two bounded continuous functions.
\end{Assumption}
\begin{Remark}
\begin{enumerate}
\item We choose to take a defaultable claim which depends only on the default time of the firm B. However, we could have been take a defaultable claim which depends on the default time of the firm A too. The calculus would have been longer but the results would have been the same.
\item Moreover, taking a defaultable claim depending only on the default time of the firm B has an economic sense. Indeed, our traded asset is the defaultable bond of the firm A, so it is justified to take payoff $g$ and $f$ function of $D^A$. Therefore if we see the firm B as an insurance company which covers the firm A, then the default of B means the counterparty default risk.
\end{enumerate}
\end{Remark}
\begin{Proposition}\label{Value}Under Assumption \ref{HypClaim}, the value function of the dual problem $\eqref{pbdual2}$ is given by:
\begin{equation}\label{Valuefunction}
V(t,D^A_t,H^A_t,H^B_t):=\inf_{\rho^A,\rho^B \in \mathcal C}\mathbb{E}^{\mathbb{Q}}\left[\int_t^Tj(s,\rho^A_s,\rho^B_s,D^A_s)ds-\delta g(D^A_T)\textbf{1}_{\{\tau^B> T\}}\Big\vert D^A_t,H^A_t,H^B_t \right],
\end{equation} 
where the function $j$ is defined by:
\begin{equation}\label{exprj}
\begin{split}
j(s,\rho^A_s,\rho^B_s,D^A_s)&=\sum_{i\in\{A,B\}}\lambda^i_s\left[(1+\rho^i_s)\ln(1+\rho^i_s)-\rho^i_s\right] -\delta(1+\rho^B_s)\lambda^B_sf(D^A_s)+{1\over 2}\rho^2_s.
\end{split}
\end{equation} 
\end{Proposition}
\begin{proof}
The proof is based on the It\^o's formula. The dynamics of $\ln(Z^{\mathbb{Q}})$ under $\mathbb{Q}$ is given by 
$$
d\ln(Z^{\mathbb{Q}}_t)=\sum_{i\in\{A,B\}}\rho^i_tdM^i_t+\left[\ln(1+\rho^i_t)-\rho^i_t\right]dH^i_t+\rho_tdW_t-{1\over 2}\rho^2_tdt.
$$
Using Girsanov theorem, the processes defined for all $i\in\{A,B\}$ by 
$$
\widetilde M^i_t=M^i_t-\int_0^t \rho^i_s\lambda^i_sds \quad and \quad  \widetilde W_t=W_t-\int_0^t \rho_s ds
$$
 are $\mathbb{Q}$-martingales. Hence, we obtain that
\begin{eqnarray*}
\hspace{-1cm}\ln(Z^{\mathbb{Q}}_T)-\delta\psi=\hspace{-0.1cm}\int_0^T \hspace{-0.3cm}\sum_{i\in\{A,B\}}\lambda^i_t[(1+\rho^i_t)\ln(1+\rho^i_t)-\rho^i_t]dt-\hspace{-0.1cm}\delta\left[\int_0^T f(D^A_{t^-})dH^B_t+g(D^A_T)(1-H^B_T)\right]+\hspace{-0.2cm}\int_0^T\hspace{-0.2cm}{1\over 2}\rho^2_t dt+\hspace{-0.15cm}M^{\mathbb{Q}}_T
\end{eqnarray*}
\noindent where $M^{\mathbb{Q}}$ is a $\mathbb{Q}$-martingale.  Then, we can rewrite the dual problem using the last expression:
\begin{equation*}
\begin{split}
\inf_{\mathbb{Q}\in {\mathcal M}^e}\mathbb{E}^{\mathbb{Q}}\left[\ln(Z^{\mathbb{Q}}_T)-\delta \psi\right]=\inf_{\rho^A,\rho^B\in \mathcal C}\mathbb{E}^{\mathbb{Q}}\left[\int_0^T j(s,\rho^A_s,\rho^B_s,D^A_s)ds-\delta(1-H^B_T)g(D^A_T)\right],
\end{split}
\end{equation*}
\noindent where $j$ is given in $\eqref{exprj}$. Since by Remark \ref{HypMarkov}, the process $(D^A,H^A,H^B)$ is a Markov process, then using the standards results of \cite{DAVIDS}, the value function of the dual optimization problem is given by:
\begin{equation*}
\hspace{-0.4cm}V(t,D^A_t,H^A_t,H^B_t)=\hspace{-0.4cm}\inf_{\rho^A,\rho^B\in \mathcal C}\hspace{-0.3cm}\mathbb{E}^{\mathbb{Q}}\left[\int_t^Tj(s,\rho^A_s,\rho^B_s,D^A_s)ds-\delta g(D^A_T)\textbf{1}_{\{\tau^B> T\}}\Big\vert D^A_t,H^A_t,H^B_t \right].
\end{equation*} 
\end{proof}
\begin{Proposition}\label{Optistrat} Let $z=(x,h^A,h^B)$ and $h=(h^A,h^B)$, then the value function of the dual optimization problem is solution of the following Hamilton-Jacobi-Bellman equation:
\begin{equation}\label{HJB}
\begin{split}
\hspace{-1cm}{\partial V\over \partial t}(t,z)+{1\over 2}{\partial V\over \partial x^2}(t,z)\sigma^2(t,z)+\hspace{-0.4cm}\inf_{\rho^A,\rho^B \in{\mathcal C}}\hspace{-0.4cm}\left\{{\mathcal L}_{\rho^A,\rho^B}V(t,z)+j(t,\rho^A_t,\rho^B_t)\right\}=0,\quad\quad V(T,z)=g(x)(1-h^B)
\end{split}
\end{equation}
\noindent where 
$${\mathcal L}_{\rho^A,\rho^B}V(t,z)= \sum_{i\in \{A,B\}} \left[ -{\partial V\over \partial z}(t,z)\sigma^i(t,z)+\left(V(t,z^i)-V(t,z)\right)\right](1+\rho^i_t)\lambda^i(t,h),$$
and $z^i=\left(x(1+\sigma^i(t,z)),h^A+\alpha^i,h^B+1-\alpha^i\right)$ where  $\alpha^A=1$ and $\alpha^B=0$. 
\noindent Moreover given the value function, the optimal strategy satisfies:
\begin{equation*}
\begin{split}
\pi^*_t=-\frac{1}{\delta}\left({\partial V\over \partial x}(t,z)+{\bar\rho_t\over D^A_{t^-}\sigma(t,z)}\right)
\end{split}
\end{equation*}
\noindent where the process $\bar \rho$ is explicitly given with the optimal control $\bar \rho^i$, $i\in\{A,B\}$, see the relation $\eqref{constraint}$.  
\end{Proposition}
\begin{proof}
From Proposition \ref{Value}, we find that the value function of the dual optimization problem is given by $\eqref{Valuefunction}$. Since $\left(D^A,H^A,H^B\right)$ is Markovian under $\P$ and that the risk neutral probability measure $\Q$ depends on the control $(\rho^A,\rho^B)$, we can apply the same method as in \cite{DAVIDS} section 3.2 and 3.3. So, using now Hamilton-Jacobi-Bellman (HJB) equation we get:
$$V(t,D^A_t,H^A_t,H^B_t)=\inf_{\rho^A,\rho^B \in {\mathcal C}}\mathbb{E}^{\mathbb Q}\left[\int_t^{t+h}j(s,\rho^A_s,\rho^B_s,D^A_s)ds+V(t+h,H^A_{t+h},H^B_{t+h})\big\vert D^A_t,H^A_t,H^B_t\right].$$
\noindent 
Then the value function solve the HJB equation $\eqref{HJB}$.\\
We are now interesting in finding the optimal strategy given the value function. Let recall that from Theorem \ref{DeSch}, the optimal risk-neutral probability and the value function exist. Let define $\bar \rho^A,\bar\rho^B$ and $\bar\rho$ be the optimal density parameters. Since $\bar \rho^A$ and $\bar \rho^B$ are optimal for the HJB equation, assuming $\sigma(t,z)\not=0$ , using first order condition, we find for $i\in \{A,B\}$:
\begin{eqnarray}\label{optpar}
\left[(V(t,z^i)\hspace{-0.1cm}-\hspace{-0.1cm}V(t,z))\hspace{-0.1cm}-\hspace{-0.1cm}x\sigma^i(t,z){\partial V\over \partial x}(t,z)\hspace{-0.1cm}+\hspace{-0.1cm}\ln(1+\bar \rho^i_t)\hspace{-0.1cm}-\hspace{-0.1cm}{\sigma^i(t,z)\over \sigma(t,z)}\bar \rho_t\right]\hspace{-0.1cm}\lambda^i(t,h)\hspace{-0.15
cm}=\hspace{-0.1cm}\delta (1-\alpha^i)f(x)\lambda^i(t,h).\hspace{0.3cm}
\end{eqnarray}
Then using the HJB equation \eqref{HJB} and the relation $\eqref{optpar}$, we obtain the following relation:
\begin{eqnarray}\label{optpar2}
-{1\over 2}\bar\rho^2_t+\sum_{i\in \{A,B\}}\bar\rho^i_t\lambda^i(t,h)=\sum_{i\in \{A,B\}}(1+\bar\rho^i_t){\sigma^i(t,z)\over\sigma(t,z)}\bar \rho_t+{1\over 2}{\partial^2 V\over\partial x^2}(t,z)x^2\sigma^2(t,z)+{\partial V\over \partial t}(t,z).\hspace{0.5cm}
\end{eqnarray}
\noindent Let recall the Ito's decomposition of the process $\ln(Z^{\mathbb{Q}^*})$:
\begin{equation*}
\ln(Z^{\mathbb{Q}^*}_T)=\int_0^T [\bar\rho_td\bar W_t+{1\over 2}\bar \rho^2_tdt]+\int_0^T \sum_{i\in \{A,B\}}\left[\ln(1+\bar\rho^i_t)dH^i_t-\bar\rho^i_t\lambda^i(t,h)\right].
\end{equation*}
\noindent Then using  equations $\eqref{optpar}$ and $\eqref{optpar2}$, we get an useful and more explicit decomposition of the process $\ln(Z^{\mathbb{Q}^*}_T)$:
\begin{equation*}
\begin{split}
\ln(Z^{\mathbb{Q}^*}_T)=&\int_0^T -{1\over 2}{\partial^2 V\over \partial x^2}(t,z_t){(D^A_{t^-})}^2\sigma^2(t,z_t)dt-\int_0^T{\partial V\over \partial t}(t,z_t)dt+\int_0^T \bar\rho_td\bar W_t\\& -\sum_{i\in \{A,B\}}\left[(V(t,z^i_t)-V(t,z_t))-D^A_{t^-}\sigma^i(t,z_t){\partial V\over \partial x}(t,z)\right]dH^i_t\\&
+\int_0^T\sum_{i\in \{A,B\}} {\sigma^i(t,z_t)\over \sigma(t,z_t)}\bar\rho_t[dH^i_t-(1+\bar \rho^i_t)\lambda^i(t,h_t)]+\int_0^T \delta f(D^A_{t^-})dH^B_t,
\end{split}
\end{equation*}                                                                                                                                            
\noindent where $z_t=(D^A_t,H^A_t,H^B_t)$ and $h_t=(H^A_t,H^B_t)$. The It\^o's decomposition of $V(T,D^A_T,H^A_T,H^B_T)$ gives:
\begin{equation*}
\begin{split}
\ln(Z^{\mathbb{Q}^*}_T)&= \int_0^T {\bar \rho_t\over \sigma(t,z_t)}\left[\sigma(t,z_t)d\bar W_t+\sum_{i\in \{A,B\}} \sigma^i(t,z_t)d\bar M^i_t\right]+\delta f(D^A_{{\tau^B}^-})\textbf{1}_{\{\tau^B\le T\}}\\&-V(T,D^A_T,H^A_T,H^B_T)+V(0,D^A_0,H^A_0,H^B_0)+\int_0^T {\partial V\over \partial x}(t,z_t)dD^A_t.
\end{split}
\end{equation*}                                                                                                                    
\\\\\noindent Since 
$$
V\left(T,D^A_T,H^A_T,H^B_T\right)=-\delta g(D^A_T)(1-H^B_T)
$$ 
and 
$$
\psi=f(D^A_{{\tau^B}^-})\textbf{1}_{\{\tau^B\le T\}}+g(D^A_T)(1-H^B_T),
$$ 
we get:  
$$
\ln(Z^{\mathbb{Q}^*}_T)-\delta \psi=V(0,D^A_0,H^A_0,H^B_0)+ \int_0^T \left[{\bar \rho_t\over D^A_{t^-}\sigma(t,z)}+{\partial V\over \partial x}(t,z_t)\right]dD^A_t.$$                                                                                                                                    
Finally, from the definition of the value function,  we have
$$
V(0,D^A_0,H^A_0,H^B_0)=\mathbb{E}^{{\mathbb Q}^*}\left[\ln(Z^{\mathbb Q^*}_T)-\delta \psi\right],
$$
and using the fact that $\mathbb{E}^{{\mathbb Q}^*}\left[X^{x,\pi^*}_T\right]=x-{1\over \delta}\ln(c)$ where $X^{x,\pi^*}_T=-{1\over \delta}\ln\left({1\over\delta} {Z^{{\mathbb Q}^*}_T\over Z^{{\mathbb P}^\psi}}\right)$ (see Theorem \ref{DeSch}), we deduce that
$$
\mathbb{E}^{{\mathbb Q}^*}\left[-{1\over \delta}\ln(Z^{{\mathbb Q}^*}_T)+\psi-{1\over \delta}\ln(c)-{1\over\delta }\ln\left({y\over \delta}\right)\right]=x-{1\over \delta}\ln(c).
$$ 
Hence, we conclude 
$$
V(0,D^A_0,H^A_0,H^B_0)=-\delta x-\ln\left(y\over \delta\right),
$$ 
and we find
\begin{equation*}
-{1\over \delta} \ln(Z^{\mathbb{Q}^*}_T)+\psi=x+{1\over \delta}\ln\left(y\over \delta\right)+ \int_0^T -{1\over \delta}\left[{\bar \rho_t\over D^A_{t^-}\sigma(t,z)}+{\partial V\over \partial x}(t,z_t)\right]dD^A_t.
\end{equation*}                                                                                                                                            
Therefore from equation $\eqref{EqSolD1}$, we obtain the expected result.
\end{proof}
  In conclusion, we found that since we can characterize  the optimal probability for the dual optimization problem using Kramkov and Schachermayer Theorem, we can characterize the HJB equation solution of our dual problem. This allows us to find the optimal strategy for the primal solution for a defaultable contingent claim $\psi$. Therefore we can find for $\psi=0$ and $\psi\not=0$ the optimal strategy and deduce the indifference price $p$ of a defaultable contingent claim solving the equation $u^\psi(x+p)=u^0(x)$.
\section{Generalization of the hedging in a general framework: Mean-Variance approach}
In this part, we assume that we work in a more general setting (not necessarily Markov), then we cannot use the HJB equation to characterize the corresponding value function. To solve our problem, we will use the Mean Variance approach. It is a well-known methodology, introduced by Schweizer in \cite{Sch92}, to manage hedging in general case. 
An important quantity in this context is the Variance Optimal Martingale Measure (VOM). The VOM, $\mathbb{\bar P}$, is the solution of the dual problem of minimizing the $L^2$-norm of the density $\frac{d\mathbb{Q}}{d\mathbb{P}}$, over all (signed) local martingale measure $\mathbb{Q}$ for $D^A$. Let recall now the Mean-Variance problem:
\begin{equation}\label{MeanVariance}
V(x)=\min_{\pi\in {\mathcal A}}\mathbb{E}\left[{(X^{x,\pi}_T-\psi)}^2\right].
\end{equation}
If we assume that $\mathbb{G}=\mathbb{F}$ (in this case we do not consider jump of default), then the process $D^A$ is continuous. In this case Delbean and Schachermayer in \cite{DS96} proved the existence of an equivalent VOM $\mathbb{\bar P}$ with respect to $\mathbb{P}$ and the fact that the price of $\psi$ is given by $\mathbb{E}^{\mathbb{\bar P}}(\psi)$. In the discontinuous case, when the so-called Mean-Variance Trade-off process (MVT) (see \cite{Sch92} for definition) is deterministic, Arai \cite{Arai05} proved the same results. Since we work in discontinuous case and since the Mean Variance Trade-off process is not more deterministic (due to the stochastic default intensity process), we cannot apply the standard results.
\begin{Remark}\label{RemHYPMOD}
Indeed, in this part we do not assume anymore that intensity processes $\lambda^A$ and $\lambda^B$ to be deterministic. We take general stochastic default intensity processes. But we assume that default times $\tau^A$ and $\tau^B$ are ordered, $\tau^A<\tau^B$ and that the ($\mathcal{H}$)-hypothesis 
 holds. A financial interpretation of this assumption could be the counterparty risk. Indeed, the firm A could be a bank (counterparty) and the firm B its company insurance which covers its default. 
\end{Remark}
So our work is firstly to characterize the value process of the Mean-Variance problem using system of BSDE's. Secondly, to make some links with the existence and the characterization of the VOM in some particular cases. Thirdly, to prove the existence of the solution of each BSDE and to give a verification Theorem. 

We begin by recalling some usual spaces:

\medskip
$\bullet $ For $s\le T$, $\mathcal{S}^\infty[s,T] $  is the Banach space of $\R$-valued cadlag
processes $X$ such that there exists a  constant $C$ satisfying
$$ \|X\|_{\mathcal{S}^\infty[s,T]}:=   \sup_{t \in[s, T]}|X_t|   \, \leq C \, <  + \infty.$$
$\bullet$  For $s\le T$, $\mathcal{H}^{2}[s,T]$   is the Hilbert space of $\mathbb{R}$-valued
predictable processes $Z$ such that
$$ \|Z\|_{\mathcal{H}^{2}[s,T]}:= \left(\mathbb{E} \, \Big[\int_s^T |Z_t|^2 \, dt\Big] \right)^{\frac{1}{2}} \, < \,   +\infty .
$$
$\bullet$  $\rm{BMO}$ is the space of $\mathbb{G}$-adapted matingale such that for any stopping times $0\le \sigma\le \tau\le T$, there exists a non negative constant $c>0$ such that:
$$\mathbb{E}\left[ [M]_{\tau}- [M]_{\sigma^-}\vert {\mathcal G}_{\sigma}\right]\le c,$$
then $M=Z.W \in {\rm BMO}$, to simplify notation we write $Z\in {\rm BMO}$.
\begin{Definition}[$R_2(\mathbb{P})$ condition]  Let $Z$ be a uniformtly integrable martingale with $Z_0=1$ and $Z_T>0$, we say that $Z$ satisfies the reverse H\"older condition $R_2(\mathbb{P})$ under $\mathbb{P}$ if there exists a constant $c>0$ such that for every stopping times $\sigma$, we have:
$$\mathbb{E}\left[{\left({Z^2_T\over Z^2_\sigma}\right)}^2\vert {\mathcal G}_\sigma\right]\le c.$$ 
\end{Definition}
\subsection{Characterization of the optimal cost via BSDE}
On our problem of \textit{mean-variance hedging (MVH)} \eqref{MeanVariance}, the performance of an admissible trading strategy $\pi \in \Ac$  is measured over the finite horizon T for an initial capital $x>0$ by
\begin{equation} \label{EqJ0}   
J^\psi (T,\pi) = \E[(X^{x,\pi}_T-\psi)^2].
\end{equation}
We use the dynamic programming principle to solve our mean variance hedging problem. Let first denote by $\Ac(t,\nu)$ the set of controls coinciding with $\nu$ until time $t\in[0,T]$
\begin{equation}\label{EqAnu}
\Ac(t,\nu) = \{ \pi \in \Ac : \pi_{.\wedge t} = \nu_{.\wedge t}\}.
\end{equation}
We can now define, for all $t\in [0,T]$, the dynamic version of \eqref{EqJ0} which is given by
\begin{eqnarray}\label{VProblem}
J^\psi(t,\pi)=\essinf_{\pi \in \Ac(t, \nu)} \E \left[ \left(X_T^{\psi,\pi}-\psi\right)^2 \vert \Gc_t\right].
\end{eqnarray}
Let recall now the dynamic programming principle given by El Karoui in \cite{ElK81}. 
\begin{Theorem}\label{dynamic} Let ${\cal S}$ be the set of  $\G$-stopping times.
\begin{enumerate}
\item The family $\{ J^\psi(\tau,\nu),\tau \in{\cal S}, \nu\in {\cal A}\}$ is a submartingale system, this implies that for any $\nu \in {\cal A}$, we have for any $\sigma\le \tau$, the following submartingale property:
\begin{equation}\label{eqsub}
\E\left[J^\psi(\tau,\nu^0)\vert {\cal G}_\sigma\right]\ge J^\psi(\sigma,\nu),\quad \P-a.s.
\end{equation} 
\item $\nu^* \in {\cal A}$ is optimal if and only if $\{ J^\psi(\tau,\nu^*),\tau \in{\cal S}\}$ is a martingale system, this means that instead of \eqref{eqsub}, we have for any stopping times $\sigma\le \tau$ that:
\begin{equation*}
\E\left[J^\psi(\tau,\nu^*)\vert {\cal G}_\sigma\right]= J^\psi(\sigma,\nu^*),\quad \P-a.s.
\end{equation*} 
\item For any $\nu \in{\cal A}$, there exists an adapted RCLL process $J^{\psi}(\nu)={(J^{\psi}(\nu)_t)}_{0\le t\le T}$ which is right closed submartingale such that:
$$J^\psi_\tau (\nu)= J^\psi(\tau,\nu), \P-a.s, \hbox{ for any stopping time } \tau.$$ 
\end{enumerate}
\end{Theorem}
\noindent We search as in Lim \cite{LimA04} a quadratic decomposition form for $J^\psi_t$ as
\begin{equation}\label{MeanVariance2}
J^\psi_t (\pi)={\Theta}_t{\left(X^{x,\pi}_t-Y_t\right)}^2+\xi_t
\end{equation}
such that $\Theta$ is a non-negative $\mathbb{G}$-adapted process and $Y,\xi$ are two $\mathbb{G}$-adapted processes. So, we assume the quadratic form \eqref{MeanVariance2} of the cost conditional $J^\psi$ with respect to the wealth process and we will use the Theorem \ref{dynamic} to characterize the triple $(\Theta,Y,\xi)$ as solution of three BSDEs. We will verify in the section \ref{Verification} that the assumption of the quadratic decomposition form, the optimality and admissibility of the founded optimal strategy are satisfied.

So, let $\pi \in \Ac$ be an admissible strategy, by representation Theorem \ref{ThmReprez}, we have that the triplet $(\Theta,Y,\xi)$ need to satisfies the following BSDEs: 
\begin{equation}\label{TripleBSDEs}
\begin{split}
&{d\Theta_t\over \Theta_{t^-}}=-g^1_t(\Theta_t,\theta^A_t,\theta^B_t,\beta_t)dt+\theta_t^AdM^A_t+\theta_t^BdM^B_t+\beta_tdW_t, \quad \hspace{1.4cm} \Theta_T=1\\
&dY_t=-g^2_t(Y_t,U^A_t,U^B_t,Z_t)dt+U^A_tdM^A_t+U^B_tdM^B_t\hspace{0.2cm}+Z_tdW_t,\quad \hspace{1cm} Y_T=\psi\\
&d\xi_t=-g^3_t(\xi_t,\epsilon^A_t,\epsilon^B_t,R_t)dt+\epsilon^A_tdM^A_t+\epsilon^B_tdM^B_t+R_tdW_t,\quad \hspace{1.8cm}\xi_T=0.
\end{split}
\end{equation}
with the constraint that $\Theta_t\ge \delta>0$, for some non negative constant $\delta$, for all $t\in [0,T].$ The processes $\theta^A,\theta^B,U^A,U^B,\epsilon^A$ and $\epsilon^B$ are $\mathbb{G}$-predictable. Hence, we can use It\^o's formula and integration by part for jump processes to find the decomposition of $J^\psi(\pi)$. Let recall that for any semimartingale $S$ and $L$, we have that  
 $$
 d (S_tL_t)=S_{t^-}dL_t+L_{t^-}dS_t+d[S,L]_t.
 $$
 In our framework since a jump comes from defaults events we get 
 $$
 d[S,L]_t=\langle S^c,L^c\rangle_t+\sum_{i\in\{A,B\}} \Delta S^i_t \Delta L^i_t dH^i_t.
 $$
 Applying these results for $S=L=(X^{x,\pi}-Y)$ gives:
\begin{equation*}
\begin{split}
d{(X^{x,\pi}-Y)}^2_t&=2(X^{x,\pi}_{t^-}-Y_{t^-})\left[(\pi_t\mu_t+g^2_t)dt+\sum_{i\in\{A,B\}}(\pi_t\sigma^i_t-U^i_t)dM^i_t+(\pi_t\sigma_t-Z_t)dW_t\right]\\&+{(\sigma_t\pi_t-Z_t)}^2 dt+\sum_{i\in\{A,B\}}{(\pi_t\sigma^i_t-U^i_t)}^2dH^i_t.
\end{split}
\end{equation*}
Secondly take $S=\Theta$ and $L={(X^{x,\pi}-Y)}^2$ , let define $K:=(X^{x,\pi}-Y)$, we find:   
\begin{equation*}
\begin{split}
d\left(\Theta K^2\right)_t&=2 K_{t^-}\Theta_{t^-}\left[(\pi_t\mu_t+g^2_t)dt+\sum_{i\in\{A,B\}}(\pi_t\sigma^i_t-U^i_t)dM^i_t+(\pi_t\sigma_t-Z_t)dW_t\right]\\&+\Theta_{t^-}{(\sigma_t\pi_t-Z_t)}^2 dt+\sum_{i\in\{A,B\}}\Theta_{t^-}{(\pi_t\sigma^i_t-U^i_t)}^2dH^i_t-\Theta_{t^-}K^2_{t^-}g^1_tdt\\
&+\Theta_{t^-}K^2_{t^-}\left[\sum_{i\in\{A,B\}}\theta^i_tdM^i_t+\beta_tdW_t\right]+2 K_{t^-}\Theta_{t^-}(\pi_t\sigma_t-Z_t)\beta_tdt\\&+\sum_{i\in\{A,B\}}\left[{(\pi_t\sigma^i_t-U^i_t)}^2+2K_{t^-}{(\pi_t\sigma^i_t-U^i_t)}\right]\theta^i_t\Theta_{t^-}dH^i_t.
\end{split}
\end{equation*}
\noindent Using this decomposition, we can write explicitly the dynamics of $J^\psi(\pi)$ for any $\pi\in{\mathcal A}$, $dJ^\psi_t(\pi)=dM^\pi_t+dV^\pi_t$ where $M^\pi_t$ is the martingale part and $V^\pi_t$ the finite variation part of $J^\psi_t$: 
\begin{equation}\label{dynamicsJ}
\begin{split}
dJ^\psi_t(\pi)=dM^\pi_t+\Theta_{t^-}\left[\pi^2_ta_t+2\pi_t(b_tK_t+c_t)+2K_t(g^2_t-u_t)-K^2_tg^1_t+v_t\right]dt-g^3_tdt
\end{split}
\end{equation}
\noindent where processes are defined respectively by:
\begin{equation}\label{Valueprocess}
\begin{split}
&a_t=\sigma^2_t+\sum_{i\in\{A,B\}}{(\sigma^i_t)}^2(1+\theta^i_t)\lambda^i_t>0, \quad b_t=\mu_t+\sigma_t\beta_t+\sum_{i\in\{A,B\}}\sigma^i_t\theta^i_t\lambda^i_t,\\
&c_t=-\sigma_tZ_t-\sum_{i\in\{A,B\}}\sigma^i_tU^i_t(1+\theta^i_t)\lambda^i_t,\quad v_t=Z^2_t+\sum_{i\in\{A,B\}}{(U^i_t)}^2(1+\theta^i_t)\lambda^i_t,\\
&\quad \textrm{and} \quad u_t=\beta_tZ_t+\sum_{i\in\{A,B\}}U^i_t\theta^i_t\lambda^i_t.
\end{split}
\end{equation}
Using now Theorem \ref{dynamic}, we have that, for any $\pi\in \Ac$, the process $J^\psi(\pi)$ is a submartingale and that there exists a startegy $\pi^*\in \Ac$ such that $J^\psi(\pi^*)$ is a martingale. This martingale property implies that we should find $\pi^*$ such that the finite variation part of $J^\psi(\pi^*)$ vanishes. 
Since the coefficients $g^1,g^2$ and $g^3$ do not depend on the strategy $\pi$, using the first order condition, we obtain:
\begin{equation}\label{EqPiStar}
\pi^*_t=-{b_tK_t+c_t\over a_t},\quad t\le T \quad \textrm{where} \quad K_t=X^{x,\pi^*}_t-Y_t. 
\end{equation}
Therefore, substituting this explicit expression of the optimal strategy in $\eqref{dynamicsJ}$, we obtain:
\begin{equation*}
\begin{split}
dJ^\psi_t(\pi)&=dM^{\pi^*}_t+\Theta_{t^-}\left[-{{(b_tK_t+c_t)}^2\over a_t}+2K_t(g^2_t-u_t)-K^2_tg^1_t+v_t\right]dt-g^3_tdt\\
&=dM^{\pi^*}_t+\Theta_{t^-}\left[-K^2_t\left(g^1_t+{b_t^2\over a_t}\right)+2K_t\left(g^2_t-u_t-{b_tc_t\over a_t}\right)\right]dt+\left((v_t-{c^2_t\over a_t})\Theta_{t^-}-g^3_t\right)dt.
\end{split}
\end{equation*}
Then setting $g^1_t+{b^2_t\over a_t}=0$, $g^2_t-u_t-{b_tc_t\over a_t}=0$ and $(v_t-{c^2_t\over a_t})\Theta_{t^-}-g^3_t=0$, we find that our coefficients $g^1,g^2$ and $g^3$ are given by:
\begin{equation*}
\begin{split}
&g^1_t(\Theta_t,\theta^A_t,\theta^B_t,\beta_t)=-{{\left[ \mu_t+\sum_{i\in\{A,B\}}\theta^i_t\sigma^i_t\lambda^i_t+\sigma_t\beta_t\right]}^2\over \sigma^2_t+\sum_{i\in\{A,B\}} (1+\theta^i_t){(\sigma^i_t)}^2\lambda^i_t},\\
\\\vspace{1.5cm}
&g^2_t(Y_t,U^A_t,U_t^B,Z_t)=-{{\left[ \mu_t+\sum_{i\in\{A,B\}}\theta^i_t\sigma^i_t\lambda^i_t+\sigma_t\beta_t\right]}\left[\sigma_tZ_t+\sum_{i\in\{A,B\}}(1+\theta^i_t)\sigma^i_tU^i_t\lambda^i_t\right]\over \sigma^2_t+\sum_{i\in\{A,B\}} (1+\theta^i_t){(\sigma^i_t)}^2\lambda^i_t},\\&\hspace{2.4cm}+\sum_{i\in\{A,B\}} \theta^i_tU^i_t\lambda^i_t+\beta_tZ_t
\\\\\vspace{1.5cm}
&g^3_t(\xi_t,\epsilon^A_t,\epsilon^B_t,R_t)=\Theta_{t^-}\left[Z^2_t+\sum_{i\in \{A,B\}}{(U^i_t)}^2(1+\theta^i_t)\lambda^i_t-{{\left(Z_t\sigma_t+\sum_{i\in \{A,B\}}\sigma^i_tU^i_t(1+\theta^i_t)\lambda^i_t\right)}^2\over \sigma^2_t+\sum_{i\in\{A,B\}} (1+\theta^i_t){(\sigma^i_t)}^2\lambda^i_t}\right] .
\end{split}
\end{equation*}
Moreover the solution of the optimization problem \eqref{MeanVariance} follows the quadratic form:
\begin{equation}
V(x)=\Theta_0{(x-Y_0)}^2+\xi_0.
\end{equation}

\noindent We are now interesting in the proof of the existence of the solution of each BSDE.

\begin{Remark}\label{Remarkxi} (Existence of the third BSDE)
\begin{enumerate}
\item If we find the solution of the first BSDE $(\Theta,\theta^A,\theta^B,\beta) \in {\mathcal S}^{\infty}[0,T]\times {\mathcal S}^{\infty}[0,T]\times {\mathcal S}^{\infty}[0,T]\times \mathrm{BMO}$, with the constraint $\Theta\ge \delta>0$ and the second BSDE $(Y,U^A,U^B,Z)\in {\mathcal S}^{\infty}[0,T]\times {\mathcal S}^{\infty}[0,T]\times {\mathcal S}^{\infty}[0,T]\times \mathrm{BMO}$ then the solution of the third is given by:
$$\xi_t=\mathbb{E}\left[\int_t^T \left(\left(v_s-{c^2_s\over a_s}\right)\Theta_{s}\right) ds \Big\vert {\mathcal G}_t\right], \quad t\le T.$$
Then $|\xi|\in {\mathcal S}^{\infty}$ and from representation Theorem \ref{ThmReprez}, we deduce that the martingale part $M$ of $\xi$:
$$M_t=\int_0^t \sum_{i\in \{A,B\}}\epsilon^i_s dM^i_s+\int_0^t R_s dW_s$$
is $\mathrm{BMO}$. Moreover from Lemma \ref{jump}, $\epsilon^A$ and $\epsilon^B$ are bounded. Therefore $(\xi,\epsilon^A,\epsilon^B,R)\in {\mathcal S}^{\infty}[0,T]\times {\mathcal S}^{\infty}[0,T]\times {\mathcal S}^{\infty}[0,T]\times \mathrm{BMO}$
\item In the complete market case, we have that the \textit{tracking error} $\xi \equiv0$ since the hedging is perfect.
\end{enumerate}
\end{Remark}

Now, we give the Theorem which proves the existence of the solution of the first quadratic BSDE.

\begin{Theorem}\label{existence} There exists a vector $(\Theta,\theta^A,\theta^B,\beta)\in {\mathcal S}^{\infty}[0,T]\times {\mathcal S}^{\infty}[0,T]\times  {\mathcal S}^{\infty}[0,T]\times{\rm BMO}$ solution of the quadratic BSDE
$${d\Theta_t\over \Theta_{t^-}}=-g^1_t(\Theta_t,\theta^A_t,\theta^B_t,\beta_t)dt+\theta^A_tdM^A_t+\theta^B_tdM^B_t+\beta_tdW_t, \quad \hspace{1.4cm} \Theta_T=1.$$
Moreover there exists a non negative constant $\delta>0$ such that $\Theta_t\ge \delta$ for all $t\in [0,T]$.  Given $(\Theta,\theta^A,\theta^B,\beta)$, we can prove the existence of $(Y,U^A,U^B,Z)\in {\mathcal S}^{\infty}[0,T]\times {\mathcal S}^{\infty}[0,T] \times {\mathcal S}^{\infty}[0,T]\times \mathrm{BMO}$ solution of the second BSDE $(g^2,\psi)$ and $(\xi,\epsilon^A,\epsilon^B, R) \in {\mathcal S}^{\infty}[0,T]\times {\mathcal S}^{\infty}[0,T]\times {\mathcal S}^{\infty}[0,T]\times \mathrm{BMO}$ solution to the third BSDE $(g^3,0)$. Moreover, given this triplet solution $(\Theta,Y,\xi)$ of our system of BSDEs \eqref{TripleBSDEs}, the solution of the our optimization problem \eqref{MeanVariance} is given by:
$$V(x)=\Theta_0{(x-Y_0)}^2+\xi_0.$$
\end{Theorem}
\noindent The proof of this Theorem will be given in the sequel in section \ref{SecTHm3.5}.

\subsection{Verification Theorem}\label{Verification} 
Given the solution of the triple BSDEs in their respective spaces (Theorem \ref{existence}), we have to verify that the assertions defined in Theorem \ref{dynamic} hold true (i.e. the submartingale and martingale properties of the cost functional $J$ are satisfied and the strategy $\pi^*$ defined in \eqref{EqPiStar} is admissible). Moreover, we will prove that the wealth process associated to $\pi^*$ exists (satisfies a stochastic differential equation (SDE)).

We begin by proving the existence of a solution of the SDE for the wealth process associated to $\pi^*$.
\begin{Proposition}
Let $\pi^*$ be the strategy, given by \eqref{EqPiStar}, then there exists a solution of the following SDE:
\begin{equation}\label{EDSopti}
dX^{x,\pi^*}_t=\pi^*_t\left[\mu_t dt +\sigma^A_t dM^A_t+\sigma^B_t dM^B_t +\sigma_t dW_t\right] \quad \textrm{with}\quad X^{x,\pi^*}_t=x.
\end{equation}
Moreover, $\pi^*$ is admissible (i.e. $\pi^* \in \Ac$).
\end{Proposition}
\begin{proof}The proof is divided in three steps. Firstly, we prove the existence of a solution of the SDE satisfied by the wealth process associated to $\pi^*$. Secondly, we prove the squared integrability of this wealth at the horizon time $T$. Finally, we prove the admissibility of the strategy $\pi^*$.
\begin{description}
\item[The existence  of the solution of the SDE for the wealth process: ]
Plotting the expression of $\pi^*$ given by \eqref{EqPiStar} in \eqref{EDSopti} gives
\begin{equation}\label{eqX0}
\hspace{-2cm}dX^{x,\pi^*}_t=(\bar b_t X^{x,\pi^*}_t+\bar c_t)dt+(\bar d^A_t X^{x,\pi^*}_t+\bar e^A_t)dM^A_t+(\bar d^B_t X^{x,\pi^*}_t+\bar e^B_t)dM^B_t+(\bar d_t X^{x,\pi^*}_t+\bar e_t)dW_t 
\end{equation}
where the bounded processes are given by
\begin{equation*}
\begin{split}
 &\bar b_t=-{b_t\over a_t}\mu_t,\quad \bar c_t=({b_t {Y_t\over a_t}+c_t})\mu_t,\quad  \bar d_t=-{b_t\over a_t}\sigma_t \\
 &\bar e_t=({b_t {Y_t\over a_t}+c_t})\sigma_t,\quad \bar d^i_t=-{\bar b_t\over a_t}\sigma^i_t, \quad \bar e^i_t=({b_t {Y_t\over a_t}+c_t})\sigma^i_t,
\end{split}
\end{equation*}
and processes $a$, $b$ $c$ are defined in \eqref{Valueprocess}. We recall, now, that the solution of the SDE:
$$d\phi_t=\phi_{t^-}\left[\bar b_t dt+\bar d^A_tdM^A_t+\bar d^B_t dM^B_t+\bar d_t dW_t\right] \quad \textrm{with}\quad \phi_0=x$$ 
is given explicitely by
$$\phi_t=x\exp\left(\int_0^t \left(\bar b_s-{1\over 2}{\bar d_s}^2-\sum_{i\in \{A,B\}}d^i_s \lambda^i_s\right)ds+\int_0^t \bar d_s dW_s\right)\prod_{i\in \{A,B\}}(1+d^i_t H^i_t).$$ 
\noindent Setting $X^{x,\pi^*}_t:= L_t\phi_t$ with 
$$dL_t:=q_tdt+l^A_t dM^A_t +l^B_t dM^B_t+l_t dW_t,\quad L_0=1,$$ 
we find by integration by part formula that $dX^{x,\pi^*}_t=\phi_{t^-} dL_t +L_{t^-}d\phi_t +d[\phi,L]_t$. Hence,
\begin{equation*}
\begin{split}
dX^{x,\pi^*}_t&
=X^{x,\pi^*}_t\left[\bar b_t dt+\bar d^A_tdM^A_t+\bar d^B_t dM^B_t+\bar d_t dW_t\right]+\phi_{t^-}\left[q_t-\sum_{i\in\{A,B\}}d^i_tl^i_t\lambda^i_t\right]dt\\
&+\sum_{i\in\{A,B\}}\phi_{t^-} l^i_t(1+d^i_t)dM^i_t+\phi_{t^-} l_tdW_t+l_t\phi_{t^-} \bar d_t dt.
\end{split}
\end{equation*}
Therefore from equation \eqref{eqX0}, we find, for $i\in \{A,B\}$, that $\bar e^i_t=\phi_{t^-}l^i_t(1+d^i_t)$, $\bar e_t=\phi_{t^-}l_t$ and $\bar c_t=\phi_{t^-}\left(q_t-\sum_{i\in\{A,B\}}d^i_tl^i_t\lambda^i_t\right)$. We deduce that the process $L$ is defined by:
\begin{equation*}
L_t=1+\int_0^t {1\over \phi_{s^-}}\left[\bar c_s+\sum_{i\in \{A,B\}}{d^i_se^i_s\over (1+d^i_s)}\lambda^i_s\right]ds+\int_0^t {\bar e_s\over \phi_{s^-}}dW_s +\int_0^t \sum_{i\in\{A,B\}} {1\over \phi_{s^-}}{e^i_s\over (1+d^i_s) }dM^i_s,
\end{equation*}
and $X^{x,\pi^*}_t=\phi_t L_t$ is a solution of the SDE \eqref{EDSopti}.
\item[Squared integrability of the strategy $\pi^*$:] Let us prove first that $X^{x,\pi^*}\in {\mathcal H}^2[0,T]$ and $ X^{x,\pi^*}_T \in \mathrm{L}^2(\Omega,{\mathcal G}_T) $. We recall that $J^\psi_t(\pi^*)=\Theta_t{(X_t^{x,\pi^*}-Y_t)}^2+\xi_t$ is a local martingale. Therefore, there exists a sequence of localizing times ${(T_i)}_{i\in \mathbb{N}}$ for $J^\psi_t$ such that for $t\le s\le T$
$$\mathbb{E}\left[ J^\psi_{t\wedge T_i}(\pi^*)\right]=\Theta_0{(x-Y_0)}^2+\xi_0.$$ 
From Remark \ref{Remarkxi}, we have:
$$\mathbb{E}\left[\xi_{t\wedge T_i}-\xi_0\right]=-\mathbb{E}\left[\int_0^{t\wedge T_i}\left(v_s-{c^2_s\over a_s}\right)\Theta_s ds\right],\quad t\le T,$$
\noindent where $v$, $c$ and $a$ are defined in Proposition \ref{TripleBSDEs}. Since $a>0$, we have:
$$\mathbb{E}\left[\Theta_{t\wedge T_i}{(X^{x,\pi^*}_{t\wedge T_i}-Y_{t\wedge T_i})}^2\right]\le \Theta_0{(x-Y_0)}^2+\mathbb{E}\left[\int_0^{t\wedge T_i}v_s \Theta_s ds \right].$$
Moreover, since there exists a constant $\delta>0$ such that $\Theta_t>\delta$ and the process $v$ is non negative, we can apply Fatou lemma and we find when i goes to infinity that
$$\delta\mathbb{E}\left[{(X^{x,\pi^*}_{t}-Y_{t})}^2\right] \le \mathbb{E}\left[\Theta_{t}{(X^{x,\pi^*}_{t}-Y_{t})}^2\right]\le \Theta_0{(x-Y_0)}^2+\mathbb{E}\left[\int_0^{t} v_s \Theta_s ds \right].$$
Therefore $Z$ is $\mathrm{BMO}$ and the process $\theta^i, U^i \in {\mathcal S}^{\infty}[0,T]$ for $i=\{A,B\}$. We conclude $v\in {\mathcal H}^2[0,T]$. Hence, we have:
\begin{equation*}
X^{x,\pi^*} -Y \in {\mathcal H}^2[0,T]\quad \textrm{and} \quad X^{x,\pi^*}_T -Y_T \in \mathrm{L}^2(\Omega,{\mathcal G}_T).
\end{equation*}  
Since $Y\in {\mathcal S}^{\infty}[0,T]$, then we get the expected results: $X^{x,\pi^*} \in {\mathcal H}^2[0,T]$ and $X^{x,\pi^*}_T \in \mathrm{L}^2(\Omega,{\mathcal G}_T)$. 
\item[Admissibility of the strategy $\pi^*$: ] Let now prove that the strategy $\pi^*\in {\mathcal H}^2[0,T]$. Applying It\^o's formula to ${(X^{x,\pi^*})}^2$, we get $d{(X^{x,\pi^*}_t)}^2=2 X^{x,\pi^*}_{t^-}dX^{x,\pi}_t +d[X^{x,\pi^*}]_t$, then there exists a sequence of localizing times ${(T_i)}_{i\in \mathbb{N}}$ such that for all $t\le s\le T$:
\begin{equation}\label{eqttt}
\hspace{-1cm}x^2+\mathbb{E}\left[\int_0^{T\wedge T^i}{|\pi^*_s|}^2[\sigma^2_s +{(\sigma^A)}^2\lambda^A_s +{(\sigma^B)}^2\lambda^B_s]ds\right]\le \mathbb{E}\left[{(X^{x,\pi^*}_{T\wedge T_i})}^2\right]-2\mathbb{E}\left[\int_0^{T^\wedge T_i}\pi^*_s\mu_s X^{x,\pi^*}_s ds \right].
\end{equation}
\noindent Setting $K^\sigma_s=\sigma^2_s +{(\sigma^A)}^2\lambda^A_s +{(\sigma^B)}^2\lambda^B_s$ ($K^\sigma$ is the so called mean variance trade-off process), since the processes $\sigma$, $\sigma^i, \lambda^i$ are bounded, there exists a constant $K$ such that $K^\sigma\ge K$. Then, we obtain
$$-2 \pi^*_s\mu_s X^{x,\pi^*}_s\le {2\over K}{|X^{x,\pi^*}_s|}^2 {|\mu_s|}^2+{K\over 2}{|\pi^*_s|}^2,\quad 0\le s\le T.$$ 
Therefore, combining this inequality with \eqref{eqttt} gives
$$x^2+\mathbb{E}\left[\int_0^{T\wedge T^i}{|\pi^*_s|}^2 K^\sigma_s ds\right]\le \mathbb{E}\left[{(X^{x,\pi^*}_{T\wedge T_i})}^2\right]+\mathbb{E}\left[\int_0^{T\wedge T_i} {2\over K}{|X^{x,\pi^*}_s|}^2 {|\mu_s|}^2 ds \right]+ {K\over 2}\mathbb{E}\left[\int_0^{T\wedge T_i}{|\pi^*_s|}^2 ds \right].$$
Applying Fatou's lemma, when i goes to infinity, we get:
$$x^2+\mathbb{E}\left[\int_0^{T}{|\pi^*_s|}^2 K^\sigma_s ds\right]\le \mathbb{E}\left[{(X^{x,\pi^*}_{T})}^2\right]+\mathbb{E}\left[\int_0^{T} {2\over K}{|X^{x,\pi^*}_s|}^2 {|\mu_s|}^2 ds \right]+ {K\over 2}\mathbb{E}\left[\int_0^{T}{|\pi^*_s|}^2 ds \right].$$
Therefore since $K^\sigma\ge K$, we finally obtain:
$${K\over 2}\mathbb{E}\left[\int_0^{T}{|\pi^*_s|}^2 ds \right]\le \mathbb{E}\left[{(X^{x,\pi^*}_{T})}^2-x^2\right]+\mathbb{E}\left[\int_0^{T} {2\over K}{|X^{x,\pi^*}_s|}^2 {|\mu_s|}^2 ds \right].$$
\noindent Since $\mu$ is bounded, $X^{x,\pi^*}\in {\mathcal H}^2[0,T]$ and $X^{x,\pi^*}_T\in \mathrm{L}^2(\Omega,{\mathcal G}_T)$, we conclude $\pi^*\in {\mathcal H}^2[0,T]$, so $\pi^*$ is admissible. Note that this condition implies that $X^{x,\pi^*}\in {\mathcal S}^2[0,T]$ since all coefficients of the asset are bounded.
\end{description}
\end{proof}
We now prove the submartingale and the martingale properties of the cost functional.
\begin{Proposition}
For any $\pi \in {\mathcal A}$, the process $J^\psi(\pi)$ is a true submartingale and a martingale for the strategy $\pi^*$ given by \eqref{EqPiStar}. Moreover the strategy $\pi^*$ is optimal for the minimization problem \eqref{MeanVariance}. 
\end{Proposition}
\begin{proof} Firstly, we prove the submartingale and the martingale property of the cost functional then secondly we prove that the strategy $\pi^*$ is optimal. 
\begin{description}
\item[First step:] Let recall that for any $\pi\in {\mathcal A}$, the process $J^\psi(\pi)$ is a local submartingale and for $\pi^*$, $J^\psi(\pi^*)$ is a local martingale. Therefore, there exists a localizing increasing sequence of stopping times ${(T_i)}_{i\in \mathbb{N}}$ for $J^\psi$ such that for $t\le s\le T$:
\begin{equation}\label{subJ}
J^\psi_{t\wedge T_i}(\pi)\le \mathbb{E}\left[J_{s\wedge T_i}(\pi)\vert {\mathcal G}_t\right]\quad \hbox{ and } J^\psi_{t\wedge T_i}(\pi^*)=\mathbb{E}\left[J_{s\wedge T_i}(\pi^*)\vert {\mathcal G}_t\right]\quad \textrm{for any} \quad \pi \in {\mathcal A}.
\end{equation}
Moreover, for any $\pi\in {\mathcal A}$,  $J^\psi_t(\pi)=\Theta_t(X^{x,\pi}_t-Y_t)+\xi_t$ where $\Theta$ , $Y$ and $\xi$ are uniformly bounded and $X^{x,\pi}\in  {\mathcal S}^2[0,T]$. Hence, taking the limit in \eqref{subJ} when $i$ goes to infinity and applying dominated convergence Theorem, allow us to conclude.  
\item[Second step:] For any $\pi\in {\mathcal A}$, we have from the submartingale property of $J^\psi(\pi)$ and the martingale property of $J^\psi(\pi^*)$ that :
$$ \mathbb{E}\left[J^\psi_T(\pi)\right]\le J^\psi_0(\pi)=\Theta_0{(x-Y_0)}^2+\xi_0=\mathbb{E}\left[J^\psi_T(\pi^*)\right].$$
Finally, $\pi^*$ is the optimal strategy for the minimization problem \eqref{MeanVariance}.
\end{description}
\end{proof}
\subsection{ Characterization of the VOM using BSDEs}\label{VOM}
Theorem \ref{existence} leads us to construct the VOM in some complete and incomplete markets. We can find also the price of the defaultable contingent claim $\psi$ via the VOM. We consider three different cases: 
\begin{enumerate}[i.]
\item Complete market (where we assume $\mathbb{G}=\mathbb{F}$ and $\mathbb{G}=\mathbb{H}^A$).
\item Incomplete market (where we consider only the case $\mathbb{G}=\mathbb{F}\vee\mathbb{H}^A$).
\item Incomplete market (where we consider the case $\mathbb{G}=\mathbb{F}\vee\mathbb{H}^A\vee\mathbb{H}^B$).
\end{enumerate}
\begin{Remark}
\begin{itemize}
\item The case iii. corresponds to the more general case where the model depends on the market information (i.e. the filtration $\F$) and the defaults informations of the firms $A$ and $B$. Indeed, in this set up, the model depends on the default time of both firms. An economic interpretation of this case is a market with two firms where $A$ is the main firm and $B$ its insurance company. Then, the main firm $A$ could make default and cause a default of its insurance. Then it is what we call a counterparty risk.
\item The case ii. corresponds to a particular case where our model depends only on the default time of the firm $A$. In fact, in this set up, the model depends on the market's information and the possible default of the firm $A$. It can be view as a particular case of iii. with condition $\tau^B=\infty$ (i.e. no possible default of firm B).
\item The first case in i., if $\mathbb{G}=\mathbb{F}$,  corresponds to a model which depends only on the market information and not to the possible default of firms $A$ and $B$. In a economic point of view, it is a simple model without default. In the second case, $\mathbb{G}=\mathbb{H}^A$, the model depends only on the possible default of the firm $A$ and non more on the information given by the market.
\end{itemize}
\end{Remark}
\begin{Remark}
 We have explicit solution of the VOM with respect to the process $\Theta$ in the first two cases.
 \end{Remark}
\subsubsection{ Complete market} 
If we assume that $\mathbb{G}=\mathbb{F}$ (we do not consider the default impact of firms A and B on the asset's dynamics of the firm A) or $\mathbb{G}=\mathbb{H}^A$ (we do not consider the market noise) then our financial market is complete. Hence, the VOM is the unique risk-neutral probability and its dynamics can be found explicitly. Our goal in this part is then to find the solution of the triple BSDEs given the VOM $\mathbb{\bar P}$.
\begin{Proposition}\label{completeVOM} Let $\mathbb{\bar P}$ be the VOM (the unique risk-neutral probability) and let $\bar Z_T$ be the Radon Nikodym density of $\mathbb{\bar P}$ with respect to $\mathbb{P}$ on ${\mathcal G}_T$. We denote $ \bar Z_t=\mathbb{E}\left[\bar Z_T\vert {\mathcal G}_t\right]$, then for all $t\leq T$, we have that
$$\Theta_t={\bar Z^2_t\over \mathbb{E}\left[\bar Z^2_T\vert {\mathcal G}_t\right]}.$$
Moreover, for all $t \in [0,T]$, we have that $Y_t= \mathbb{\bar E}\left[\psi \vert {\mathcal G}_t\right]$ and $\xi_t\equiv 0$.
\end{Proposition}
\begin{proof}We will consider the two cases $\mathbb{G}=\mathbb{F}$ and $\mathbb{G}=\mathbb{H}^A$.
\begin{description}
\item[First case:] Let consider the case where $\mathbb{G}$ is equal to $\mathbb{F}$ and let the process $L$ be defined by the stochastic differential equation given by $dL_t=L_{t^-}\rho_t dW_t$ where $\rho W \in {\rm BMO}$. Using It\^o's formula we obtain:
\begin{equation*}
\begin{split}
d\left({L^2_t\over \Theta_t}\right)&={L^2_t\over \Theta_t}\left[(2\rho_t-\beta_t)dW_t+(\beta^2_t+g^1_t-2\beta_t\rho_t+\rho^2_t)dt\right]\\
&={L^2_t\over \Theta_t}\left[(2\rho_t-\beta_t)dW_t+\left({(\beta_t-\rho_t)}^2-{({\mu_t\over \sigma_t}+\beta_t)}^2\right)dt\right]\\
&={L^2_t\over \Theta_t}\left[(2\rho_t-\beta_t)dW_t+\left((-\rho_t-{\mu_t\over \sigma_t})(2\beta_t+\rho_t+{\mu_t\over \sigma_t})\right)dt\right].
\end{split}
\end{equation*}
Then, if we set, for all $t \leq T$, $\rho_t:=-{\mu_t\over \sigma_t}$ and used the bound conditions of $\left({1\over \Theta},\mu,\sigma\right)$ and the BMO property of $\beta$, we obtain that the process ${L^2\over \Theta}$ is a true martingale. Therefore we get:
 $$\mathbb{E}\left({L^2_T\over \Theta_T}\Big\vert {\mathcal G}_t\right)={L^2_t\over \Theta_t},\quad t\le T.$$
 Since $\Theta_T=1$, we find the expected result. Moreover, we obtain that $L=\bar Z$ which is the Radon-Nikodym of the unique risk-neutral probability and $g^2_t=-{\mu_t\over \sigma_t}Z_t$ , $g^3_t=0$. Finally, $Y_t=\mathbb{\bar E}\left[\psi\vert {\mathcal G}_t\right]$ and $\xi_t=0, t\le T$.
\item[Second case:] Let now consider the case where $\mathbb{G}$ is equal to $\mathbb{H}$ and let the process $L$ be defined by the stochastic differential equation given by 
$$
dL_t=L_{t^-}\rho^A_t dM^A_t,
$$
where $\rho^AM^A \in {\rm BMO}$. Using It\^o's formula we find:
\begin{equation*}
\begin{split}
d\left({L^2_t\over \Theta_t}\right)&={L^2_{t^-}\over \Theta_{t^-}}\left[\left({{(1+\rho^A_t)}^2\over 1+\theta^A_t}-1\right)dM^A_t+\left({({(\theta^A_t)}^2+{(\rho^A_t)}^2-2\rho^A_t\theta^A_t)\lambda^A_t\over 1+\theta^A_t}+g^1_t\right)dt\right]\\
&={L^2_{t^-}\over \Theta_{t^-}}\left[\left({{(1+\rho^A_t)}^2\over 1+\theta^A_t}-1\right)dM^A_t+{1\over 1+\theta^A_t}\left({(\rho^A_t-\theta^A_t)}^2-{({\mu_t\over\sigma^A_t\lambda^A_t}+\theta^A_t)}^2\right)\lambda^A_tdt\right]\\
&={L^2_{t^-}\over \Theta_{t^-}}\left[\left({{(1+\rho^A_t)}^2\over 1+\theta^A_t}-1\right)dM^A_t+{1\over 1+\theta^A_t}\left((\rho^A_t+{\mu_t\over \sigma^A_t\lambda^A_t})(-2\theta^A_t+\rho^A_t-{\mu_t\over\sigma^A_t\lambda^A_t})\right)\lambda^A_tdt\right],
\end{split}
\end{equation*}
then if we set for all $t \leq T$
$$\rho^A_t:=-{\mu_t\over \lambda^A_t\sigma^A_t},
$$
and using the bound condition of $\Theta,\mu,\sigma^A,\theta^A$, the process ${L^2\over \Theta}$ is a true martingale. Hence we get:
 $$\mathbb{E}\left({L^2_T\over \Theta_T}\Big\vert {\mathcal G}_t\right)={L^2_t\over \Theta_t},\quad t\le T.$$
Since $\Theta_T=1$, we find again the expected result. Moreover $L=\bar Z$, $g^2_t=-{\mu_t\over \lambda^A_t}U^A_t$ ,  and $g^3_t=0$. Finally, $Y_t=\mathbb{\bar E}\left[\psi\vert {\mathcal G}_t\right]$ and $\xi_t=0, t\le T.$
\end{description}
\end{proof}
\begin{Remark}\label{proofexistence} We have proved that we can find the existence of solution of the triple BSDEs using only the VOM.
\end{Remark}
\subsubsection{ Incomplete market}
In the incomplete market case, the remark \ref{proofexistence} does not hold true. The VOM depends on the dynamics of $(\Theta,\theta^A,\theta^B,\beta)$. In the particular case where $\mathbb{G}=\mathbb{F}\vee\mathbb{H}^A$, we can find that the Proposition \ref{completeVOM} holds true. But in the more general case $\mathbb{G}=\mathbb{F}\vee\mathbb{H}^A\vee \mathbb{H}^B$, we can not prove the existence of the VOM but we can still characterize the process $\Theta$ with some martingale measure.
\begin{Proposition}\label{incompleteVOMM}
 Let consider the incomplete market $\mathbb{G}=\mathbb{F}\vee\mathbb{H}^A$, then the VOM $\bar \P$ defines the local martingale measure $\mathbb{Q}$ which minimizes the $L^2$-norm of $Z^{\mathbb{Q}}$, $\bar{Z}_T$ represents the Radon Nikodym density of $\bar{\P}$ with respect to $\mathbb{P}$ on ${\mathcal G}_T$ and  $ \bar{Z}_t=\mathbb{E}\left[\bar{Z}_T\vert {\mathcal G}_t\right]$. We find, for all $t\le T$, that 
$$\Theta_t={\bar Z^2_t\over \mathbb{E}\left[\bar Z^2_T\vert {\mathcal G}_t\right]}.$$
Moreover $$Y_t= \mathbb{\bar E}\left[\psi \vert {\mathcal G}_t\right].$$ 
In the more general case, where $\mathbb{G}=\mathbb{F}\vee \mathbb{H}^A\vee \mathbb{H}^B$, we can only prove that there exists a martingale measure $\mathbb{\bar P}$ such that for all  $t\le T$:
 $$\Theta_t={\bar Z^2_t\over \mathbb{E}\left[\bar Z^2_T\vert {\mathcal G}_t\right]}
\quad and \quad Y_t= \mathbb{\bar E}\left[\psi \vert {\mathcal G}_t\right].$$
\end{Proposition}
\begin{proof}
\begin{description}
\item[First step:]
Consider the case  where $\mathbb{G}=\mathbb{F}\vee \mathbb{H}^A$ and $\mathbb{Q}$ be a martingale measure for the asset $D^A$. Let define $Z^{\mathbb{Q}}_T$ its Radon Nikodym density with respect to $\mathbb{P}$ on ${\mathcal G}_T$. We define the process $Z^{\mathbb {Q}}_t=\mathbb{E}\left[Z^{\mathbb{Q}}_T\vert{\mathcal G}_t\right]$. Using martingale representation Theorem \ref{ThmReprez}, there exists two $\mathbb{G}$-predictable processes $\rho^A$ and $\rho$ such that 
$$
dZ^\Q_t=Z^\Q_{t^-}\left[\rho^A_tdM^A_t+\rho_tdW_t\right].
$$ 
Using It\^o's formula, we find:
\begin{equation}\label{star1}
d\left({{(Z^{\mathbb{Q}}_t)}^2\over \Theta_t}\right)={{(Z^{\mathbb{Q}}_{t^-})}^2\over \Theta_{t^-}}\left[\left({{(1+\rho^A_t)}^2\over 1+\theta^A_t}-1\right)dM^A_t+(2\rho_t-\beta_t)dW_t +j_tdt\right],
\end{equation}
where $j_t={(\rho_t-\beta_t)}^2+{{(\rho^A_t-\theta^A_t)}^2\over 1+\theta^A_t}\lambda^A_t+g^1_t$. Since $\mathbb{Q}$  is a  martingale measure for $D^A$ we get using \eqref{constraint} that
$$
\mu^A_t+\rho^A_t\sigma^A_t\lambda^A_t+\rho_t\sigma_t=0.
$$
Hence using this equation, we can find $\rho^A$ using $\rho$ and plotting this result on the expression of $j$. We obtain
\begin{equation*}
\begin{split}
j_t={(\rho_t-\beta_t)}^2+{{(\mu_t+\sigma_t\rho_t+\sigma^A_t\theta^A_t\lambda^A_t)}^2\over (1+\theta^A_t){(\sigma^A_t)}^2\lambda^A_t}-{{(\mu_t+\beta_t\sigma_t+\theta^A_t\sigma^A_t\lambda^A_t)}^2\over (1+\theta^A_t){(\sigma^A_t)}^2\lambda^A_t+\sigma^2_t}.
\end{split}
\end{equation*}
\noindent Let now define 
$$
\bar \rho_t=\rho_t-\beta_t,  \quad \bar a_t=\sigma^2_t+(1+\theta^A_t){(\sigma^A_t)}^2\lambda^A_t \quad and \quad \bar b_t=\mu_t+\sigma_t\beta_t+\sigma^A_t\theta^A_t\lambda^A_t,
$$
then we get:  
\begin{equation*}
\begin{split}
j_t&={1\over (1+\theta^A_t){(\sigma^A_t)}^2\lambda^A_t}\left[\bar a_t \bar \rho_t+2\bar \rho_t\bar b_t\sigma_t+{\bar b^2_t\sigma^2_t\over \bar a_t}\right]={\bar a_t\over (1+\theta^A_t){(\sigma^A_t)}^2\lambda^A_t}{\left(\bar \rho_t+{\bar b_t\sigma_t\over \bar a_t}\right)}^2>0,
\end{split}
\end{equation*}
\eqref{star1}, $j\ge 0$ and the fact that the process ${{(Z^{\mathbb{Q}})}^2\over \Theta}$ is a submartingale (since $Z^{\mathbb{Q}}$ is a martingale and ${1\over \Theta}\in {\mathcal S}^{\infty}[0,T]$), we deduce $\mathbb{E}\left[{{(Z^{\mathbb{Q}}_T)}^2\over \Theta_T}\right]\ge {{(Z^{\mathbb{Q}}_0)}^2\over \Theta_0}$, since $\Theta_T=1$ and $Z^{\mathbb{Q}}_0=1$.
Finally, we get for any martingale measure for $D^A$ that $\mathbb{E}\left[{(Z_T^{\mathbb{Q}})}^2\right]\ge {1\over \Theta_0}$. Moreover, if we set $\bar \rho_t=-{\bar b_t\sigma_t\over \bar a_t}$, then $\bar Z$ is a true martingale measure since $(\Theta,\theta^A,\theta^B\beta)\in {\mathcal S}^{\infty}[0,T]\times {\mathcal S}^{\infty}[0,T]\times{\mathcal S}^{\infty}[0,T]\times {\rm BMO}$ and $\mu, \sigma^A,\sigma^B$ are bounded (the process $b$,$a$, $\rho$ and $\rho^A$ are bounded). We call $\mathbb{\bar P}$ the  martingale measure under this condition, then $\mathbb{E}\left[\bar Z^2_T\right]= {1\over \Theta_0}$. We deduce $\mathbb{\bar P}$ is the martingale measure which minimizes the $L^2$-norm of $Z$ and $\bar \Theta_t={\bar Z^2_t\over \mathbb{E}\left[\bar Z^2_T\vert {\mathcal G}_t\right]}, t\le T$. 

Using the explicit expression of $\rho$, we find:
\begin{equation*}
\rho_t=-{\sigma_t\bar b_t\over \bar a_t}+\beta_t \quad \textrm{and} \quad \rho^A_t=-{(1+\theta^A_t)\sigma^A_t\bar b_t\over \bar a_t}+\theta^A_t.
\end{equation*}
\noindent Moreover since
\begin{equation*}
\begin{split}
g^2_t&={-\bar b_t(\sigma_t Z_t+(1+\theta^A_t)U^A_t\sigma^A_t\lambda^A_t)\over \bar a_t}+\beta_tZ_t+U^A_t\lambda^A_t\\
&=Z_t\left(-{\bar b_t\sigma_t\over \bar a_t}+\beta_t\right)+U^A_t\left(-{(1+\theta^A_t)\sigma^A_t\bar b_t\over \bar a_t}+\theta^A_t\right)\lambda^A_t\\
&=Z_t\rho_t+U^A_t\rho^A_t\lambda^A_t,
\end{split}
\end{equation*}
then we conclude that $Y_t=\mathbb{\bar E}\left[\psi\vert {\mathcal G}_t\right]$. Therefore the characterization of the price of $\psi$ (using Mean-Variance approach) and the VOM in this incomplete case are well defined using the vector $(\Theta,\theta^A,\theta^B,\beta)$ associated to the first BSDE.\\
\item[Second step:]
We consider now the more general case where $\mathbb{G}=\mathbb{F}\vee \mathbb{H}^A\vee \mathbb{H}^B$. Let consider $\mathbb{Q}$ be a martingale measure for the asset $D^A$ and let define $Z^{\mathbb{Q}}_T$ its Radon Nikodym density with respect to $\mathbb{P}$ on ${\mathcal G}_T$. We can define the process $Z^{\mathbb {Q}}_t=\mathbb{E}\left[Z^{\mathbb{Q}}_T\vert{\mathcal G}_t\right]$. Using martingale theorem representation \ref{ThmReprez} there exists $\mathbb{G}$-predictable processes $\rho^A$, $\rho^B$ and $\rho$ such that 
$$
dZ^\Q_t=Z^\Q_{t^-}\left[\rho^A_tdM^A_t+\rho^B_tdM^B_t+\rho_tdW_t\right].
$$ 
Using It\^o's formula, we find:
\begin{equation*}
\begin{split}
d\left({{(Z^{\mathbb{Q}}_t)}^2\over \Theta_t}\right)&={{(Z^{\mathbb{Q}}_{t^-})}^2\over \Theta_{t^-}}\left[\sum_{i\in\{A,B\}}\left({{(1+\rho^i_t)}^2\over 1+\theta^i_t}-1\right)dM^i_t+(2\rho_t-\beta_t)dW_t +j_tdt\right],
\end{split}
\end{equation*}
\noindent where $j_t={(\rho_t-\beta_t)}^2+\sum_{i\in\{A,B\}}{{(\rho^i_t-\theta^i_t)}^2\over 1+\theta^i_t}\lambda^i_t+g^1_t$. Since $\mathbb{Q}$ is a martingale measure for $D^A$ we get by \eqref{constraint} 
$$
\mu^A_t+\sum_{i\in\{A,B\}}\rho^i_t\sigma^i_t\lambda^i_t+\rho_t\sigma_t=0.
$$
Hence using this equation, we can find $\rho^A$ using $\rho$  and $\rho^B$ and then plotting this result on the expression of $j$. Let first recall a notation:
\begin{equation*}
\begin{split}
a_t=\sigma^2_t+\sum_{i\in \{A,B\}}(1+\theta^i_t){(\sigma^i_t)}^2\lambda^i_t \quad and \quad
b_t=\mu_t+\sigma_t\beta_t+\sum_{i\in \{A,B\}}\theta^i_t\sigma^i_t\lambda^i_t,
\end{split}
\end{equation*}
so we obtain:
\begin{equation*}
\begin{split}
C_t:=&(1+\theta^A_t){(\sigma^A_t)}^2\lambda^A_t j_t\\&={(\rho_t-\beta_t)}^2[\sigma^2_t+(1+\theta^A_t){(\sigma^A_t)}^2\lambda^A_t]+{{(\rho^B_t-\theta^B_t)}^2\over 1+\theta^B_t}\lambda^B_t\left[\sum_{i\in \{A,B\}}(1+\theta^i_t){(\sigma^i_t)}^2\lambda^i_t\right]\\&+{b^2_t\over a_t}\left[\sigma^2_t+(1+\theta^B_t){(\sigma^B_t)}^2\lambda^B_t\right]+2(\rho^B_t-\theta^B_t)(\rho_t-\beta_t)\sigma_t\sigma^B_t\lambda^B_t\\&+2b_t\left[(\rho_t-\beta_t)\sigma_t+(\rho^B_t-\theta^B_t)\sigma^B_t\lambda^B_t\right].
\end{split}
\end{equation*}
Then from the two first terms, we add and remove  an additional process to find the process $a$. We get:
\begin{eqnarray*}
C_t&=&\left[{(\rho_t-\beta_t)}^2a_t+{b^2_t\over a_t}\sigma^2_t+2 b_t(\rho_t-\beta_t)\sigma_t\right]+(1+\theta^B_t)\lambda^B_t\Big[{{(\rho^B_t-\theta^B_t)}^2\over {(1+\theta^B_t)}^2}a_t
+2b_t\sigma^B_t{\rho^B_t-\theta^B_t\over 1+\theta^B_t}\\&&+{b^2_t\over a^2_t}{(\sigma^B_t)}^2\Big]+(1+\theta^B_t)\lambda^B_t\left[2(\rho_t-\beta_t){(\rho^B_t-\theta^B_t)\over (1+\theta^B_t)}\sigma_t\sigma^B_t -{(\rho_t-\beta_t)}^2{(\sigma^B_t)}^2-{{(\rho^B_t-\theta^B_t)}^2\over {(1+\theta^B_t)}^2}\sigma^2_t\right].
\end{eqnarray*}
Finally, we find a more explicit expression of $C$:
\begin{eqnarray*}
C_t&=&a_t\left[{\left((\rho_t-\beta_t)+{b_t\over a_t}\sigma_t\right)}^2+(1+\theta^B_t)\lambda^B_t{\left({(\rho^B_t-\theta^B_t)\over 1+\theta^B_t}+{b_t\sigma^B_t\over a_t}\right)}^2\right]\\
&&-(1+\theta^B_t)\lambda^B_t{(\sigma^B_t)}^2 {\left((\rho_t-\beta_t)-{\sigma_t\over \sigma^B_t}{\rho^B_t-\theta^B_t\over 1+\theta^B_t}\right)}^2.
\end{eqnarray*}
It follows that if we set $\rho_t-\beta_t:=-{b_t\over a_t}\sigma_t$ and $\rho^B_t-\theta^B_t:= -(1+\theta^B_t)\sigma^B_t{b_t\over a_t}$, then we find $j=0$ and  $\rho^A_t-\theta^A_t=-\sigma^A_t{b_t\over a_t}$. Since $({1\over \Theta},\theta^A,\theta^B,\beta)\in {\mathcal S}^{\infty}[0,T]\times {\mathcal S}^{\infty}[0,T]\times {\mathcal S}^{\infty}[0,T]\times {\rm BMO}$ and $\mu$, $\sigma^A$, $\sigma^B$ and $\sigma$ are bounded then the processes $b$, $a$, $\rho$, $\rho^A$ and $\rho^B$ are bounded too. Therefore, we deduce that there exists a martingale measure $\mathbb{\bar P}$ such that 
\begin{equation}\label{ineqVom} 
{\delta}\le \Theta_t={\bar Z^2_t\over \mathbb{E}\left[\bar Z^2_T\vert {\mathcal G}_t\right]},\quad t\le T.
\end{equation}
Moreover we find, for all $t \le T$, that 
$$
g^2_t=Z_t\rho_t+\sum_{i\in\{A,B\}}U^i_t\rho^i_t\lambda^i_t
$$ 
then 
$$Y_t=\mathbb{\bar E}\left[\psi \vert {\mathcal G}_t\right].$$
\end{description}
\end{proof}
\begin{Remark}\label{VM}(About the VOM)
\begin{itemize}
\item To identify that $\mathbb{\bar P}$ is the VOM in the general case where $\mathbb{G}=\mathbb{F}\vee \mathbb{H}^A\vee \mathbb{H}^B$, we should prove that $j\ge 0$ (as in the first case of the previous Proposition). But from the last expression of $j$, we can not prove that this condition holds true. However, we can remark that the assertion of VOM will be justify if one of the following equality is satisfied:
\begin{equation*}
 \sigma^B_t(\rho_t-\beta_t)=\sigma_t{\rho^B_t-\theta^B_t\over 1+\theta^B_t},\quad \sigma^A_t{\rho^B_t-\theta^B_t\over 1+\theta^B_t}=\sigma^B_t{\rho^A_t-\theta^A_t\over 1+\theta^A_t}\quad \textrm{or} \quad\sigma^A_t(\rho_t-\beta_t)=\sigma_t{\rho^A_t-\theta^A_t\over 1+\theta^A_t}  .
\end{equation*}
\item The generalization of the expectation under a $\sigma$-measure ( $Y_t=\mathbb{\bar E}\left[\psi \vert {\mathcal G}_t\right])$ was defined by Cerny and Kallsen in \cite{CeKa07} p 1512. Moreover, given the solution of the first BSDE: $(\Theta,\theta^A,\theta^B,\beta) \in {\mathcal S}^{\infty}[0,T]\times {\mathcal S}^{\infty}[0,T]\times {\mathcal S}^{\infty}[0,T]\times \mathrm{BMO}$ with the constraint $\Theta\ge \delta>0$, the martingale $\bar Z$ satisfies \eqref{ineqVom}. Since $\psi$ is bounded, we conclude:
$$|Y_t|\le 2\mathbb{E}\left[{\bar Z^2_T\over \bar Z^2_t}\vert {\mathcal G}_t\right]+2\mathbb{E}\left[[\psi|^2\vert {\mathcal G}_t\right]\le 2\left[{1\over \delta}+||\psi||^2_{\infty}\right].$$
Therefore $Y\in {\mathcal S}^{\infty}[0,T]$ and from representation theorem \ref{ThmReprez}, the martingale part $M$ of $Y$ given by
$$M_t=\int_0^t \sum_{i\in\{A,B\}} U^i_s dM^i_s +\int_0^t Z_s dW_s$$
 is $\mathrm{BMO}$. Moreover from Lemma \ref{jump} in Appendix, since $Y\in {\mathcal S}^{\infty}[0,T]$, we obtain that $\theta^A$ and $\theta^B$ are bounded. We conclude if the solution of the first BSDE exists $(\Theta,\theta^A,\theta^B,\beta) \in {\mathcal S}^{\infty}[0,T]\times {\mathcal S}^{\infty}[0,T]\times {\mathcal S}^{\infty}[0,T]\times \mathrm{BMO}$, with the constraint $\Theta\ge \delta>0$, that the solution of the second BSDE $(Y,U^A,U^B,Z)\in {\mathcal S}^{\infty}[0,T]\times {\mathcal S}^{\infty}[0,T]\times {\mathcal S}^{\infty}[0,T]\times \mathrm{BMO}$ exists. 
\end{itemize}
\end{Remark}

\subsection{Proof of Theorem \ref{existence}}\label{SecTHm3.5}
We prove in this part the existence of $(\Theta,\theta^A,\theta^B,\beta)$ in the space ${\mathcal S}^{\infty}[0,T]\times {\mathcal S}^{\infty}[0,T]\times {\mathcal S}^{\infty}[0,T]\times \mathrm{BMO}$ with the constraint $\Theta>\delta$. Moreover, we recall that given the solution of this first BSDE, the existence of the second and the third BSDEs is given by Remark \ref{Remarkxi} and \ref{VM}.
\par\medskip
Note that to prove the existence of $(\Theta,\theta^A,\theta^B,\beta)$, we do not need the assumption that the VOM exists and should satisfied the $\R_2(\P)$ condition (this assumption implies that the Radon-Nikodym of  the VOM $\bar \P$ with respect to $\P$ on ${\mathcal G}_T$ is non-negative). Moreover, if $(\Theta,\theta^A,\theta^B,\beta)$ is defined such that $\bar Z$ is non negative, then it implies that $\bar \P$ satisfies the $R_2(\P)$ condition.
\par\medskip In fact, in the general discontinuous filtration, it is difficult to prove that we can find $(\Theta,\theta^A,\theta^B,\beta)$ solution of the first  BSDE such that $\Theta>0$ (see \cite{LimA04} for more discussions about the difficulty). Indeed in the set up of \cite{LimA04}, the author makes the hypothesis that all coefficients of the asset are $\F$-predictable. This strong hypothesis makes the jump part of process $\Theta$ equals to zero. In our framework, this hypothesis can not be satisfied since the intensities processes  are $\G$-adapted. Hence, we deal with splitting method of BSDE defined by  \cite{KhaLim11} to prove that, even if the jump part of the process is not equal to zero, we can split the jump BSDE in continuous BSDEs such that each BSDE have a solution in a good space.  The proof is divided in two parts. Firstly, we will give the splitted BSDEs in this framework and secondly, we will solve recursively each BSDE. 
\begin{description}
\item[First step:] Let define, for all $t\in [0,T]$, $\bar g_t=\Theta_{t^-} g^1_t$, $\bar \theta^i_t=\Theta_{t^-}\theta^i_t$ for $i\in\{A,B\}$, $\bar \theta_t=\bar\theta^A_t1_{\{t< \tau^A\}}+\bar\theta^B_t1_{\{\tau^A\le t\le \tau^B\}}$ and $\bar \beta_t=\Theta_{t^-}\beta_t$. Then, we can define the BSDE $(\bar g,\Theta_T)$ which is given by:
$$d\Theta_t=-\bar f_tdt+\bar\theta^A_tdH^A_t+\bar\theta^B_tdH^B_t+\bar \beta_tdW_t \quad \textrm{with} \quad \Theta_T=1,$$
where $\bar f_t=\bar g_t+\bar\theta^A_t \lambda^A_t+\bar\theta^B_t\lambda^B_t$. We also define
$$\Delta_k=\{(l_1,\cdots l_k)\in {(\R^+)}^k: l_1\le\cdots \le l_k\}, \quad 1\le k\le 2.$$
Since we work with the same assumption (density assumption) and notation as in \cite{KhaLim11}, then we can decompose $\Theta_T$ and $\bar g$ between each default events:
\begin{eqnarray*}
\Theta_T&=&\gamma^0 1_{\{\{0\le T< \tau_A\}\}}+\gamma^1(\tau^A)1_{\{\tau^A\le T\le\tau^B\}}+\gamma^2(\tau^A,\tau^B)1_{\{\tau^B<T\}}
\end{eqnarray*}
and
\begin{eqnarray}\label{eqgderive}
\bar f_t(\Theta_t,\bar\theta_t,\bar\beta_t)&=&\bar f^0_t(\Theta_t,\bar\theta_t,\bar\beta_t) 1_{\{0\le t< \tau_A\}}+\bar f^1_t(\Theta_t,\bar\theta_t,\bar\beta_t,\tau^A)1_{\{\tau^A\le t\le \tau^B\}}\\
&+&\bar f^2_t(\Theta_t,\bar\theta_t,\bar\beta_t,(\tau^A,\tau^B))1_{\{\tau^B<t\}}
\end{eqnarray}
\noindent where $\gamma^0$ is ${\cal F}_T$-measurable, $\gamma^k$ is ${\cal F}_T\otimes{\cal B}(\Delta_k)$-measurable for $k=\{1,2\}$, $\bar g^0$ is ${\cal P}(\mathbb{F})\otimes {\cal B}(\mathbb{R})\otimes {\cal B}(\mathbb{R})$-measurable and $\bar g^k$ is ${\cal P}(\mathbb{F})\otimes {\cal B}(\mathbb{R})\otimes {\cal B}(\mathbb{R})\otimes {\cal B}(\Delta^k)$. Moreover, since $\Theta_T=1$ (bounded) (see proposition 3.1 in \cite{KhaLim11}), we have that the variables $\gamma^k(l_{(k)})=1$, for $k=\{0,1,2\}$.

Let now give the main result of splitting BSDE which is a first step to prove the existence of $(\Theta,\bar\theta^A,\bar\theta^B,\bar\beta)$. Let  $l_{(2)}=(l_1,l_2) \in \Delta_2$ and assume that the following BSDE:
\begin{equation}\label{bsde1split}
d\Theta^2_t(l_{(2)})=-\bar f^2_t\left(\Theta^2_t(l_{(2)}),0,\bar\beta^2_t(l_{(2)}),l\right)dt+\bar\beta^2_t(l_{(2)})dW_t,\quad \Theta^2_T(l_{(2)})=\gamma^2(l_{(2)})
\end{equation}
 admits a  solution $\left(\Theta^2_t((l_{(k+1)})),\bar\beta^2_t((l_{(k+1)}))\right)\in S^\infty ([l_2\wedge T,T])\times {\cal H}^2[l_2\wedge T,T]$. And that
\begin{eqnarray}\label{splibsde}
d\Theta^k_t(l_{(k)})&=&-\bar f^k_t\left(\Theta^k_t(l_{(k)}),(\Theta^{k+1}_t(l_{(k)},t)-\Theta^{k}_t(l_{(k)})),\bar \beta^k_t(l_{(k)}),l_{(k)}\right)dt+\bar \beta^k_t(l_{(k)})dW_t,\nonumber\\
 \Theta^k_T(l_{(k)})&=&\gamma^k(l_{(k)})
\end{eqnarray}  
admits ,for $k=\{0,1\}$, a solution $\left(\Theta^k(l_{(k)},\bar \beta^k(l_{(k)})\right)\in  {\cal S}^\infty ([l_k\wedge T,T])\times {\cal H}^2[l_k\wedge T,T]$, where $l_{(k)}=(l_1,\cdots l_k)$.
Then $(\Theta,\bar\theta^A,\bar \theta^B,\bar\beta)$ is given following \cite{KhaLim11} by:
\begin{equation}\label{recollement}
\begin{split}
&\Theta_t=\Theta^0_t 1_{\{t<\tau^A\}}+\Theta^1_t(\tau^A)1_{\{\tau^A\le t\le \tau^B\}}+\Theta^2_t(\tau^A,\tau^B)1_{\{\tau^B<t\}},\\&
\bar \beta_t=\bar \beta^0_t 1_{\{t<\tau^A\}}+\bar \beta^1_t(\tau^A)1_{\{\tau^A\le t\le \tau^B\}}+\bar \beta^2_t(\tau^A,\tau^B)1_{\{\tau^B<t\}},\\&
\bar \theta^B_t=\Theta^{2}_t(\tau_{A},t)-\Theta^1_t(\tau_{A}),\\&
\bar \theta^A_t=\Theta^{1}_t(t)-\Theta^0_t(t).
\end{split}
\end{equation}  
 Therefore, to prove the existence of $(\Theta,\theta^A,\theta^B,\beta)$ we have to prove the existence of solutions of the BSDEs $\eqref{bsde1split}$ and $\eqref{splibsde}$. 
\item[Second step (the recursive approach):] We prove recursively the existence of the solution of these BSDEs. Firstly, we prove the existence of the solution of $\eqref{bsde1split}$ and secondly, assuming that the solution of \eqref{splibsde} exists and satisfies the constraint for step $k+1$, we prove the same assertion for the step $k$.
\begin{enumerate}
\item Let consider the BSDE $\eqref{bsde1split}$:
\begin{equation*}
d\Theta^2_t(l_{(2)})=-\bar f^2_t(\Theta^2_t(l_{(2)}),0,\bar\beta^2_t(l_{(2)}),l_{(2)})dt+\bar\beta^2_t(l_{(2)})dW_t,\quad \Theta^2_T(l_{(2)})=\gamma^2(l_{(2)}).
\end{equation*}
Since the coefficient $\bar g$ is given by:
\begin{equation}\label{gmodif}
\begin{split}
&\bar g_t(\Theta_t,\bar\theta,\beta_t)=-{{\left[ \mu_t\Theta_t+\sum_{i\in\{A,B\}}\bar\theta^i_t\sigma^i_t\lambda^i_t+\sigma_t\bar\beta_t\right]}^2\over \Theta_t\sigma^2_t+\sum_{i\in\{A,B\}} (\Theta_t+\bar\theta^i_t){(\sigma^i_t)}^2\lambda^i_t},
\end{split}
\end{equation}
from the predictable decomposition of the composition of assets coefficient:
\begin{equation}\label{recollement}
\begin{split}
&\sigma_t=\sigma^0_t 1_{\{t<\tau^A\}}+\sigma^1_t(\tau^A)1_{\{\tau^A\le t\le \tau^B\}}+\sigma^2_t(\tau^A,\tau^B)1_{\{\tau^B<t\}},\\&
\mu_t=\mu^0_t 1_{\{t<\tau^A\}}+\mu^1_t(\tau^A)1_{\{\tau^A\le t\le \tau^B\}}+\mu^2_t(\tau^A,\tau^B)1_{\{\tau^B<t\}},\\&
\sigma^A_t=\sigma^{1,0}_t 1_{\{t<\tau^A\}}+\sigma^{1,1}_t(\tau^A)1_{\{\tau^A\le t\le \tau^B\}}+\sigma^{1,2}_t(\tau^A,\tau^B)1_{\{\tau^B<t\}},\\&
\sigma^B_t=\sigma^{2,0}_t 1_{\{t<\tau^A\}}+\sigma^{2,1}_t(\tau^A)1_{\{\tau^A\le t\le \tau^B\}}+\sigma^{2,2}_t(\tau^A,\tau^B)1_{\{\tau^B<t\}},
\end{split}
\end{equation}  
and by our model assumption (see Remark \ref{RemHYPMOD}), we get
$$\lambda^A_t=\lambda^{1,0}_t 1_{\{t<\tau^A\}} \quad \textrm{and} \quad \lambda^B_t=\lambda^{2,1}_t(\tau^A)1_{\{\tau^A\le t\le \tau^B\}}.$$
We so find that:
$$\bar f^2_t(\Theta^2_t(l_{(2)}),0,\bar \beta^2_t(l_{(2)}),l_{(2)})=\Theta^2_t(l_{(2)}){\left[{\mu^2_t(l_{(2)})\over {(\sigma^2_t(l_{(2)}))}^2}+{\bar\beta^2_t(l_{(2)})\over \Theta^2_t(l_{(2)})}\right]}^2, \quad t\in [0,T]$$
\noindent Using the result of Section \ref{VOM}, in the complete market case when $\G=\F$, we conclude:
\begin{equation}\label{Theta2}
\Theta^2_t(l_{(2)})={Z^2_t(l_{(2)})\over \mathbb{E}\left[{Z^2_T(l_{(2)})\over \gamma^2(l_{(2)})}\right]}, \quad t\le T, l_{(2)}\in \Delta_2,
\end{equation}
\noindent where the family of processes $Z(.)$ satisfies the SDE given by
$$
{dZ_t(l_{(2)})\over Z_t(l_{(2)})}=-{\mu^2_t(l_{(2)})\over \sigma^2_t(l_{(2)})}dW_t,$$
with $Z_0(l_{(2)})=1$. Since $\mu^2$ and $\sigma^2$ are bounded, the martingale $M_t(l_{(2)}):=\int_0^t {\mu^2_s(l_{(2)})\over \sigma^2_s(l_{(2)})}dW_s$ is $\rm{BMO}$. We deduce so that $Z(l_{(2)})$ satisfies the $R_2(\P)$ inequality. Moreover $\gamma^2(l_{(2)})=1$, we conclude that there exists a constant $\delta^2>0$ such that for all $t\in [0,T]$ and $l_{(2)}\in \Delta_2$,  $\Theta^2_t(l_{(2)})\ge \delta^2$. The existence of  $\bar\beta^2(l_{(2)})$ is given by the martingale part of the process given by \eqref{Theta2}. Moreover since $\Theta^2(l_{(2)})$ is bounded then the coefficient $\bar g^2$ satisfies a quadratic growth with respect to $\bar\beta^2(l_{(2)})$. Therefore since the terminal condition $\gamma^2(l_{(2)})$ is bounded, we conclude from Kobylanski \cite{Kob00}, that $\bar\beta(l_{(2)})$ is $\mathrm{BMO}$.
\item Let assume now that there exists a solution which satisfies the constraint for the step $k+1$. That means that the pair $(\Theta^{k+1}(l_{(k+1)}),\bar\beta^{k+1}_t(l_{(k+1)}))\in {\mathcal S}^{\infty}[l_{k+1},T]\times \rm{BMO}$ and that there exists  a non negative constant $\delta^{k+1}$ such that $\Theta^{k+1}(l_{(k+1)})\ge \delta^{k+1}$. Let us now proved the existence of the pair $(\Theta^k(l_{(k)}),\bar \beta^k(l_{(k)}))\in {\mathcal S}^{\infty}[l_k,T]\times \mathrm{BMO}$ at step $k$:
\begin{equation*}
\begin{split}
&d\Theta^k_t(l_{(k)})=-\bar f^k_t\left(\Theta^k_t(l_{(k)}),(\Theta^{k+1}_t(l_{(k)},t)-\Theta^{k}_t(l_{(k)})),\bar \beta^k_t(l_{(k)}),l_{(k)}\right)dt+\bar \beta^k_t(l_{(k)})dW_t, \\&\Theta^k_T(l_{(k)})=\gamma^k(l_{(k)}).
\end{split}
\end{equation*}  
From the decomposition of \eqref{gmodif}, we find:
\begin{equation*}
\begin{split}
&\bar f^k_t\left(\Theta^k_t(l_{(k)}),(\Theta^{k+1}_t(l_{(k)},t)-\Theta^{k}_t(l_{(k)})),\bar \beta^k_t(l_{(k)}),l_{(k)}\right)\\&=-{{\left[ \mu^k_t(l_{(k)})\Theta^k_t(l_{(k)})+(\Theta^{k+1}_t(l_{(k)},t)-\Theta^{k}_t(l_{(k)}))\sigma^{k+1,k}_t(l_{(k)})\lambda^{k+1,k}_t(l_{(k)})+\sigma^k_t(l_{(k)})\bar\beta^k_t(l_{(k)})\right]}^2\over \Theta^k_t(l_{(k)})
{\sigma^k_t(l_{(k)})}^2+(\Theta^{(k)}_t(l_{(k)})+\Theta^{k+1}_t(l_{(k)},t)-\Theta^{k}_t(l_{(k)})){\sigma^{k+1,k}_t(l_{(k)})}^2\lambda^{k+1,k}_t(l_{(k)})}\\
&+\left[\Theta^{k+1}_t(l_{(k)},t)-\Theta^{k}_t(l_{(k)})\right]\lambda^{k+1,k}_t(l_{(k)}).
\end{split}
\end{equation*}
Let consider the processes:
\begin{equation*}
\begin{split}
&n_t=\left[\mu^k_t(l_{(k)})-\sigma^{k+1,k}_t(l_{(k)})\lambda^{k+1,k}_t(l_{(k)})\right]\vert\Theta^k_t(l_{(k)})\vert
,\\& \kappa_t:=\sigma^{k+1,k}_t(l_{(k)})\lambda^{k+1,k}_t(l_{(k)})\Theta^{k+1}_t(l_{(k)},t)
,\\& m_t:=\sigma^k_t(l_{(k)})\bar\beta^k_t(l_{(k)}).
\end{split}
\end{equation*}
and $\bar N_t= n_t +\kappa_t+m_t$, $d_t=|\Theta^k_t(l_{(k)})|{\sigma^k_t(l_{(k)})}^2$, $p_t=\Theta^{k+1}_t(l_{(k)},t){(\sigma^{k+1,k}_t(l_{(k)})}^2\lambda^{k+1,k}_t(l_{(k)})$ and $ D_t=d_t+p_t$.
\noindent 
We define
\begin{equation}\label{Eqf}
\begin{split}
f^{k}_t:&=\bar f^{k}_t(|\Theta^k_t(l_{(k)})|,\bar \beta^k_t,l_{(k)})\\&=-{{\bar N}^2_t\over \bar D_t}+\left[\Theta^{k+1}_t(l_{(k)},t)-|\Theta^{k}_t(l_{(k)})|\right]\lambda^{k+1,k}_t(l_{(k)}),
\end{split}
\end{equation}

where $\bar N^2_t=n^2_t+m^2_t+\kappa^2_t+2n_tm_t+2\kappa_tn_t+2\kappa_tm_t$. 
Since the process $\Theta^{k+1}_t(l_{(k)},t)\ge \delta^{k+1}>0$, then there exists a non negative constant $c>0$ such that $p_t>c$. Hence, we obtain:
\begin{equation*}
\begin{split}
&-f^{k}_t:={{\bar N}^2_t\over \bar D_t}-\left[\Theta^{k+1}_t(l_{(k)},t)-|\Theta^{k}_t(l_{(k)})|\right]\lambda^{k+1,k}_t(l_{(k)})\\&\le
\left[{n^2_t\over d_t}+{2n_t \kappa_t\over p_t}+|\Theta^{k}_t(l_{(k)})|\lambda^{k+1,k}_t(l_{(k)})\right]+{m^2_t\over d_t}\\&+\left[{2n_tm_t\over d_t}+{2m_t\kappa_t\over p_t}\right]+\left[{\kappa^2_t\over p_t}-\Theta^{k+1}_t(l_{(k)},t)\lambda^{k+1,k}_t(l_{(k)})\right].
\end{split}
\end{equation*}
Therefore since all processes $\Theta^{k+1}$, $\mu^k$, $\sigma^{k+1,k}$ and $\lambda^{k+1,k}$ are bounded, there exists bounded processes $a$, $b$ and $c$ such that:
$$ -f^{k}_t\le h_t:=a_t {|\Theta^k_t}(l_{(k)})|+b_t\bar\beta^k_t(l_{(k))}+{{\bar\beta^k_t(l_{(k)})^2\over |\Theta^k_t(l_{(k)})|}}.$$
The coefficient $f^k$ has a quadratic growth and from Kobylanski \cite{Kob00} since the terminal condition $\gamma^k$ is bounded, there exists a  pair $(\Theta^k(l_{(k)}), \bar\beta(l_{(k)}))$ solution of the BSDE associated to $(f^k(l_{(k)}),\gamma^k(l_{(k)}))$. Moreover, if we consider  the BSDE $d\bar x_t=-\bar h(\bar x_t,\bar Z_t)dt+\bar Z_tdW_t$ with terminal condition $x_T=\gamma^1(l_{(1)})=1$ and where the coefficient $\bar h$ is given by :
$$\bar h_t(\bar x_t,\bar Z_t)=-a_t \bar x_t-{\bar Z^2_t\over \bar x_t} -b_t \bar Z_t,\quad t\in [0,T],$$
we obtain using Proposition 5.11 in \cite{LimA04} that the solution $(\bar x, \bar Z) \in {\mathcal S}^{\infty}[0,T]\times\mathrm{BMO}$ exists. 

Moreover, there exists a non negative constant $\delta^k$ such that $\bar x_t\ge \delta^k, a.s$. Hence we conclude using Comparison theorem of quadratic BSDE (see \cite{Kob00}) that $f^k\ge -\bar h=h$. Finally, the pair of solution $(\Theta^k(l_{(k)}),\bar \beta^k(l_{(k)}))$ associated to $(f^k(l_{(k)}),\gamma^k(l_{(k)}))$ satisfies $\Theta^k_t(l_{(k)}) \ge \delta^k>0$ a.s for all $t\in [0,T],l_{(k)}\in \Delta_k$. Therefore, from \ref{Eqf}, we conclude $f^k=\bar f^k$. It follows that there exists a solution $(\Theta^k(l_{(k)}), \bar \beta^k(l_{(k)})) \in {\mathcal S}^{\infty}[l_k,T]\times \mathrm{BMO}$ associated to $(\bar f^k(l_{(k)}),\gamma^k(l_{(k)}))$ such that $\Theta^k_t(l_{(k)}) \ge \delta^k>0$ a.s for all $t\in [0,T], l_{(k)}\in \Delta_k$.
\end{enumerate}
 \end{description}
\subsection{Special case and explicit solution of the BSDE}
We conclude by giving an explicit example of our credit risk model which allow us to find explicit solution of each BSDEs. We assume $\mathbb{G}=\mathbb{F}\vee\mathbb{H}^A$ and that the parameters of the dynamics of the asset are constant before and after the default time $\tau^A$. Moreover, we assume that the intensity process is given by $\lambda_t=\lambda(1-H^A_t)$. Using theses assumptions, we find an explicit solution of the BSDE associated to $(g^1,\Theta_T)$ using the splitting approach.
\begin{Assumption}\label{constparams} The processes  $\mu, \sigma,\sigma^A, \lambda$ in \eqref{bondA} satisfy the following sassumptions:
\begin{equation*}
\begin{split}
&\mu_t=\mu(H^A_t)=\mu^01_{\{\tau^A>t\}}+\mu^11_{\{\tau^A\le t\}},\\
&\sigma_t=\sigma(H^A_t)=\sigma^01_{\{\tau^A>t\}}+\sigma^11_{\{\tau^A\le t\}},\\
& \sigma^A_t=\sigma^A(H^A_t)=\kappa 1_{\{\tau^A>t\}},\\
& \lambda_t=\lambda(H^A_t)=\lambda 1_{\{\tau^A>t\}},
\end{split}
\end{equation*}
such that $\mu^0\kappa=(\sigma^0)^2+\kappa^2\lambda$.
 \end{Assumption}
\begin{Proposition}\label{specialcase} Under Assumption \ref{constparams}, there exists a solution of the BSDE associated to $(g^1, \Theta)$ given by:
$$\Theta_t=\exp\left[-{\mu^0\over \kappa}(T-t)\right]1_{\{\tau^A>t\}}+\exp\left[-{\left({\mu^1\over \sigma^1}\right)}^2(T-t)\right]1_{\{\tau^A\le t\}},\quad t\le T.$$
\end{Proposition}
\begin{proof} Let first recall that using the splitting approach developed by \cite{KhaLim11}, we can write the BSDE before and after the default. We obtain
\begin{eqnarray*}
\Theta_t&=&\Theta^0_t 1_{\{t<\tau^A\}}+\Theta^1_t(\tau^A)1_{\{\tau^A\le t\}},\\
g^1_t&=&g^{1,0}_t 1_{\{t<\tau^A\}}+g^{1,1}_t1_{\{\tau^A\le t\}},
\end{eqnarray*}  
where $\Theta^0$ and $\Theta^1$ satisfy the following dynamics:
\begin{equation*}
\begin{split}
&-{d\Theta^0_t\over \Theta^0_t}=g^{1,0}_t(\Theta^0_t,\theta^A_t,\beta^0_t)dt-\beta^0_tdW_t+\lambda\theta^A_tdt,\quad \Theta^0_T=1,\\
& -{d\Theta^1_t(l)\over \Theta^1_t(l)}=g^{1,1}_t(\Theta^1_t(l),0,\beta^1_t(l))dt-\beta^1_t(l)dW_t,\quad \Theta^1_T(l)=1
\end{split}
\end{equation*}
with
\begin{eqnarray*}
g^{1,0}_t(\Theta^0_t,\theta^A_t,\beta^0_t)=-\frac{\left[ \mu^0+\theta^A_t\kappa\lambda+\sigma^0\beta^0_t\right]^2}{(\sigma^0)^2+(1+\theta^A_t)\kappa^2\lambda}\quad \textrm{and} \quad g^{1,1}_t(\Theta^1_t(l),0,\beta^1_t(l))=-\frac{\left[ \mu^1+\sigma^1\beta^1_t(l)\right]^2}{(\sigma^1)^2},
\end{eqnarray*}
where $l\in \Delta_1$ and $\Theta^1_t(t)-\Theta^0_t=\theta^A_t\Theta^0_t$ (see proof of Theorem \ref{existence} for more details). Using Assumption \ref{constparams}, setting $\beta^1(l)=0$, we find that $g^{1,1}_t(\Theta^1_t,0,\beta^1_t(l))=-\left({\mu^1\over \sigma^1}\right)^2$. Since $\Theta^1_T(l)=1$, then $\Theta^1(l)=\Theta^1$ and we get:
$$\Theta^1_t=\exp\left[-{\left({\mu^1\over \sigma^1}\right)}^2(T-t)\right],\quad t\le T.$$
To find the solution of the first one BSDE, we set $\beta^0=0$ and from Assumption \ref{constparams} we obtain
 $$
 \mu^0\kappa=(\sigma^0)^2+\kappa^2\lambda.
 $$
We deduce $g^{1,0}_t(\Theta^0_t,\theta^A_t,\beta^0_t)=-{\mu^0\over \kappa} -\theta^A_t\lambda.$ Therefore we find that $\Theta^0$ satisfies the dynamics:
$$-{d\Theta^0_t\over \Theta^0_t}=-{\mu^0\over \kappa}dt,\quad \Theta^0_T=1.$$
Finally, we get $\Theta^0_t=\exp\left[-{\mu^0\over \kappa}(T-t)\right]$ and we find the expected result.
\end{proof}
\section*{Appendix}
\begin{Lemma}\label{jump}  Let consider $X$ and $Y$ two $\mathbb{G}$-predictable processes such that for $i\in\{A,B\}$,  $Y_{\tau_i}=X_{\tau_i}$.
Then,  $X_t=Y_t$ on $(\tau_i \ge t)$ a.s. Moreover,  if $X_{\tau_i}\le Y_{\tau_i}$, then $X_t\le Y_t$ a.s on $(\tau_i \ge t)$.
\end{Lemma}
\begin{proof}  Assume  that $X$ and $Y$ are
bounded. If $X_{\tau_i}=Y_{\tau_i}$, then $\int_0^\infty
|X_t-Y_t|dH_t^i=0$ and $$0=\mathbb{E}\left( \int_0^\infty
|X_t-Y_t|dH_t^i\right)=\mathbb{E}\left[\int_0^\infty {|X_t-Y_t|} \lambda_t ^i
dt\right].  $$ Therefore, we have $X_t=Y_t$ on $(\tau^i\ge t)$.
Moreover, if   $X_{\tau_i}\le Y_{\tau_i}$, we consider the predictable process $V$ defined as
$V_t=Y_t1_{\{X_t\le Y_t\}}$. Then $V_{\tau^i}=Y_{\tau^i}$ and by using
the first part of the proof, we obtain $V_t=Y_t$ on $(\tau^i\ge t)$.
The general case follows.
\end{proof}

\end{document}